\documentclass{article}
\usepackage[utf8]{inputenc}
\usepackage[margin=1in]{geometry}
\usepackage{macros}
\usepackage[noend]{algpseudocode}
\usepackage{algorithm}

\usepackage{thmtools}

\newcommand{\cut}[1]{}

\newcommand{\metric}{\mathcal{M}}
\renewcommand{\dist}{\operatorname{d}}
\newcommand{\lip}[1]{\|#1\|_{\mathrm{Lip}}}
\newcommand{\range}{\Omega}
\newcommand{\bin}{\omega}
\newcommand{\lsh}{\mathcal{H}}
\newcommand{\disp}{\beta}
\renewcommand{\dim}{\operatorname{dim}}

\renewcommand{\tilde}{\widetilde}

\newenvironment{subproof}[1][\proofname]{%
  \begin{proof}[#1]%
}{%
  \end{proof}%
}

\title{Near Neighbor Search via Efficient Average Distortion Embeddings}
\author{Deepanshu Kush \\ University of Toronto \and Aleksandar Nikolov\thanks{Supported by an NSERC Discovery Grant (RGPIN-2016-06333), and a Tier II Canada Research Chair.} \\ University of Toronto \and Haohua Tang \\ University of Toronto}
\date{}

\begin{document}
\maketitle
\begin{abstract}
A recent series of papers by Andoni, Naor, Nikolov, Razenshteyn, and Waingarten (STOC 2018, FOCS 2018) has given approximate near neighbour search (NNS) data structures for a wide class of distance metrics, including all norms. In particular, these data structures achieve approximation on the order of \(p\) for \(\ell_p^d\) norms with space complexity nearly linear in the dataset size \(n\) and polynomial in the dimension \(d\), and query time sub-linear in \(n\) and polynomial in \(d\). The main shortcoming is the \emph{exponential in \(d\) pre-processing time} required for their construction.

In this paper, we describe a more direct framework for constructing NNS data structures for general norms. More specifically, we show via an algorithmic reduction that an efficient NNS data structure for a metric \(\metric\) is implied by an efficient average distortion embedding of \(\metric\) into \(\ell_1\) or the Euclidean space. In particular, the resulting data structures require only \emph{polynomial pre-processing time}, as long as the embedding can be computed in polynomial time.

As a concrete instantiation of this framework, we give an NNS data structure for \(\ell_p\) with \emph{efficient pre-processing} that matches the approximation factor, space and query complexity of the aforementioned~data structure of Andoni et al. On the way, we resolve a question of Naor (Analysis and Geometry in Metric Spaces, 2014) and provide an explicit, efficiently computable embedding of \(\ell_p\), for \(p \ge 2\), into \(\ell_2\) with (quadratic) average distortion on the order of \(p\). Furthermore, we also give data structures for Schatten-\(p\) spaces with improved space and query complexity, albeit still requiring exponential pre-processing when \(p\ge 2\). We expect our approach to pave the way for constructing efficient NNS data structures for all norms.
\end{abstract}
\newpage
\section{Introduction}

Nearest neighbor search is a fundamental problem in computational geometry. In this problem, we are given an \(n\)-point subset \(P\)  of a metric space \(\metric\) with a distance function \(\dist_\metric\), and our goal is to pre-process \(P\) into a data structure that, given a query point \(q\in \metric\), finds a point \(x \in P\) minimizing \(\dist_\metric(x,q)\). The main parameters of a nearest neighbor search data structure are
\begin{itemize}
    \item the \emph{pre-processing} time required to construct the data structure given \(P\);
    \item the \emph{space} taken up by the data structure,  in words of memory;
    \item the \emph{query time} required to answer a nearest neighbor query.
\end{itemize}
A trivial solution is to store \(P\) as a list of points and to answer queries by linear search. Ignoring the time required to compute distances, this solution takes \(\Theta(n)\) space, but also requires \(\Theta(n)\) query time, which is prohibitively large when we have a large dataset and expect to answer many queries. In some cases, it is possible to use the geometry of the metric to design data structures with much more efficient query procedures and nearly the same space requirements. For instance, Lipton and Tarjan~\cite{LiptonT80} gave a data structure for the nearest neighbor problem in the $2$-dimensional Euclidean plane with \(O(n\log n)\) pre-processing, \(O(n)\) space, and \(O(\log n)\) query time. This result has been extended to \(d\)-dimensional Euclidean space (see e.g.~\cite{Meiser93}), and other \(d\)-dimensional normed spaces, but with exponential dependence of the space and/or query time on the dimension $d$. There is no known nearest neighbor data structure for the \(d\)-dimensional Euclidean space that achieves space that is polynomial in \(n\) and \(d\), and query time that is polynomial in \(d\), and sub-linear in \(n\) (i.e., \(O(\poly(d)\cdot n^{1-\alpha})\) for some \(\alpha > 0\)).\footnote{Here and in the rest of the paper, we use the notation \(\poly(A)\) to denote the class of polynomials in the expression \(A\).}  There is, furthermore, some evidence that no such data structure exists~\cite{Williams05}.

The nearest neighbor search problem finds a multitude of applications beyond computational geometry, in areas as diverse as databases, computer vision, and machine
learning. For example, it is used to find joinable tables in publicly
available data~\cite{MillerNZCPA18}; for object recognition~\cite{Mel97}
and shape matching~\cite{BelongieMP02} in computer vision; to solve analogical
reasoning tasks~\cite{MikolovSCCD13}; in machine learning,
the $k$-Nearest Neighbors classifier is a common baseline. In these applications, often both the dataset size \(n\) and the dimension \(d\) are large, making query time linear in \(n\) or exponential in \(d\) unacceptable. It makes sense then to relax this problem in the hope of allowing for efficient data structures in the high-dimensional regime. A common relaxation is to allow returning an \emph{approximately nearest} neighbor to the query point \(q\), i.e., a point \(x\in P\) for which \(\dist_\metric(x,q) \le c \min_{y \in P}\dist_\metric(y,q)\) for some approximation factor \(c > 1\). A long and fruitful line of work, recently surveyed in~\cite{AIR18}, has shown that it is possible to construct data structures for this approximate nearest neighbor problem over certain spaces such as the \(d\)-dimensional Euclidean or Manhattan distance that 
use space \(O(n^{1+\varepsilon}\cdot \poly(d))\) and support queries in time \(O(n^{\varepsilon}\cdot \poly(d))\), for a constant \(\varepsilon < 1\) that tends to \(0\) as the approximation factor \(c\) tends to infinity.

Rather than solving the nearest neighbor search problem directly, it is more convenient to fix a scale for the distance, and work with the \((c,r)\)-near neighbor search (\((c,r)\)-NNS) problem, defined below.
\begin{definition}
In the \((c,r)\)-near neighbor search  (\((c,r)\)-NNS) problem, we are given a set of \(n\) points \(P\) in a metric space \((\metric, \dist_\metric)\), and are required to build a data structure so that, given a query point \(q\in \metric\) with the guarantee that \(\dist_\metric (x^*,q) \le r\) for some \(x^* \in P\), we can use the data structure to output a point \(x\in P\) satisfying \(\dist_\metric(x,q)\le cr\) with probability at least \(\frac23\).
\end{definition}
It was shown by Indyk and Motwani~\cite{IM98} that the approximate nearest neighbor problem can be reduced to solving \(\poly(\log n)\) instances of the \((c,r)\)-NNS problem. Therefore, we focus on the latter problem from this point onward. 

Most (but not all) efficient data structures for the NNS problem in the high-dimensional regime are based on the idea of locality sensitive hashing (LSH), introduced by Indyk and Motwani~\cite{IM98}. A locality sensitive family of hash functions is a probability distribution \(\lsh\) over random functions \(h:\metric \to \range\) such that pairs of close points are much more likely to be mapped by \(h\) to the same value than far points are. In particular, pairs of points at distance at most \(r\) get mapped to the same value with probability at least \(p_1\), while pairs of points at distance at least \(cr\) get mapped to the same value with probability at most \(p_2\), with \(p_1 > p_2\).  Indyk and Motwani showed that an LSH family implies a data structure for the \((c,r)\)-NNS problem with space \(O(n^{1+\rho}\log_{1/p_2}(n))\), and query time \(O(n^\rho\log_{1/p_2}(n))\), where \(\rho = \frac{\log(1/p_1)}{\log(1/p_2)}\)~\cite{IM98}. Moreover, they constructed LSH families for the Hamming and Manhattan (i.e., \(\ell_1\)) distances with \(\rho \approx \frac1c\). Subsequent work also showed the existence of LSH families for the Euclidean distance (i.e., \(\ell_2\)), as well as the \(\ell_p\) norms for \(1 \le p \le 2\), and improved the parameters~\cite{AI06,DIIM04}. 
The LSH definition above has the property that the distribution \(\lsh\) is independent of the dataset \(P\). Sometimes, however, data structures with better trade-offs can be constructed by allowing \(\lsh\) to depend on \(P\)~\cite{AINR14,ALRW16a,AR15}. 

Until recently, relatively little was known about the NNS problem in high dimension beyond the \(\ell_p\) spaces for \(1 \le p\le 2\), and the \(\ell^d_\infty\) space,\footnote{Recall that the \(\ell_p^d\) norm on \(\R^d\) is defined by \(\|x\|_p = (\sum_{i = 1}^d |x_i|^p)^{1/p}\) for \(1\le p < \infty\), and \(\|x\|_\infty = \max_{i=1}^d |x_i|\).} for which Indyk gave an efficient deterministic decision tree data structure with approximation \(O(\log \log d)\)~\cite{I01}. Data structures for other spaces can be constructed by reducing to these special cases via bi-Lipschitz embeddings. I.e., if for some metric \(\metric\) we can find an efficiently computable injection \(f:\metric \to \ell_2^d\) such that \(\|f(x)-f(y)\|_2 \approx \dist_\metric(x,y)\) for all \(x,y\in \metric\), then we can use NNS data structures for \(\ell_2^d\) to solve the NNS problem in \(\metric\). The best approximation factor achievable by this approach depends on the \emph{distortion} of \(f\), which measures how well \(\|f(x)-f(y)\|_2\)  approximates  \(\dist_\metric(x,y)\) in the worst case, and is defined as \(\lip{f}\cdot  \lip{f^{-1}}\), where 
\[
\lip{f} \coloneqq \sup_{x\neq y, x,y \in \metric}\frac{\|f(x)-f(y)\|}{\dist_\metric(x,y)},
\hspace{4em}
\lip{f^{-1}} \coloneqq \sup_{x\neq y, x,y \in \metric}\frac{\dist_\metric(x,y)}{\|f(x)-f(y)\|_2}
\]
are the Lipschitz constants of \(f\) and its inverse, respectively. Although this approach does yield some non-trivial results (see~\cite{AIR18} for a survey), it only produces data structures with approximation \(c \ge d^{\frac12 - \frac1p}\) even in the special case of \(\ell_p^d\) with \(p > 2\), as this is the best possible distortion achievable by a bi-Lipschitz embedding into \(\ell_2\) (see, e.g., \cite{BL00}). It is natural to ask if the best approximation achievable by an efficient NNS data structure for a metric \(\metric\) is characterized by the optimal distortion of a bi-Lipschitz embedding into \(\ell_2\). More generally, a fundamental problem in high-dimensional computational geometry is to \emph{determine the geometric properties of a metric space that allow efficient and accurate NNS data structures.}

A recent line of work showed that the answer to the first question above is negative in a strong sense, and there exist efficient NNS data structures for a wide range of metrics with approximation much better than what is implied by bi-Lipschitz embeddings into \(\ell_2^d\) or \(\ell_1^d\)~\cite{spectral,daher}. These papers give data-dependent LSH families for any \(d\)-dimensional normed space with approximation factor that is sub-polynomial in the dimension \(d\). A sample theorem is the following.
\begin{theorem}[\cite{spectral}]\label{thm:spectral-ellp}
For any \(r > 0\), \(p \ge 2\) and any \(\varepsilon\in (0,1)\), there is some \(c \lesssim \frac{p}{\varepsilon}\) such that the following holds. For any set \(P\) of \(n\) points in \(\R^d\) such that for all \(x\in \R^d\) we have \(|B_{\ell_p^d}(x,cr) \cap P| \le \frac{n}{2}\), there exists a probability distribution on axis-aligned boxes \(S\) satisfying
\begin{align*}
    &\Pr\left[\frac{n}{4} \le |S \cap P| \le \frac{3n}{4}\right]
    = 1\\
    &\|x-y\|_p \le r \implies
    \Pr[|S \cap \{x,y\}| =1] \le \varepsilon.
\end{align*}
\end{theorem}
Here, \(B_{\ell_p^d}(x,cr)=\{y\in \R^d: \|y-x\|_p \le cr\}\) is the \(\ell^d_p\)-ball of radius \(cr\) centered at \(x\). Moreover, here and in the rest of the paper, the notation \(A \lesssim B\) means that there exists an absolute constant \(C\), independent of all other parameters, such that \(A \le CB\). The notation \(A \gtrsim B\) is equivalent to \(B\lesssim A\).

Theorem~\ref{thm:spectral-ellp} gives a form of LSH: we can think of points inside \(S\) as being mapped to \(1\), and points outside mapped to \(0\). We still have the \(p_1\) condition: close points get mapped to the same value with probability at least \(1-\varepsilon\). Rather than guaranteeing that far points are less likely to be mapped to the same value, we have a data-dependent condition: if the dataset contains no dense clusters of points, then most pairs of points are mapped to different values. It is not hard to construct a randomized decision tree using Theorem~\ref{thm:spectral-ellp}, and \cite{spectral} showed how to use it to give a data structure for \((c,r)\)-NNS over \(\ell_p^d\) with approximation \(c\lesssim \frac{p}{\varepsilon}\), space \(O(n^{1+\varepsilon}\cdot\poly(d))\) and query time \(O(n^{\varepsilon}\cdot\poly(d))\). Similar results were also proved for the Schatten-\(p\) norms, which extend the \(\ell_p\) norms to matrices, and for arbitrary norms (with appropriate access to the norm ball) in~\cite{daher}.

Theorem~\ref{thm:spectral-ellp}, however, has a significant shortcoming: it does not guarantee that the distribution over axis-aligned boxes can be sampled efficiently given \(P\) as input. Indeed, the proof of the theorem in~\cite{spectral}, as well as the proofs of other similar results in~\cite{spectral,daher}, rely on a duality argument and yield sampling algorithms with running time exponential in the dimension. For this reason, the resulting data structures have pre-processing time that is also exponential in the dimension. These works thus raise an intriguing open problem: \emph{Can we sample a distribution such as the one in Theorem~\ref{thm:spectral-ellp} in time polynomial in \(n\) and \(d\)?}

\subsection{Our Results on Near Neighbor Search}
In this work, we resolve the open problem above (also posed explicitly in~\cite{spectral}), and prove the following theorem. 
\begin{restatable}{theorem}{dsellp}\label{thm:ds-ellp}
Let \(\varepsilon \in (0,1]\), \(r > 0\), \(p \ge 2\). For some \(c \lesssim \frac{p}{\varepsilon}\), there exists a data structure for the \((c,r)\)-NNS problem over \(n\)-point sets in \(\ell_p^d\) with 
\begin{itemize}
    \item pre-processing time \(\poly(nd)\);
    \item space \(O(n^{1+\varepsilon}\log(n) \cdot\poly(d))\);
    \item query time \(O(n^{\varepsilon}\log(n) \cdot\poly(d))\).
\end{itemize}
\end{restatable}
We note that the only previous NNS data structure over \(\ell_p^d\) (for \(p > 2\)) with pre-processing time \(\poly(nd)\) could only achieve approximation on the order of either \(2^p\)~\cite{NR06,BG15}, or \(\log \log d\)~\cite{Andoni09}. Very recent work has shown that an approximation factor on the order of \(O((\log p)(\log d)^{2/p})\) is achievable, but the data structures requires exponential preprocessing time~\cite{L1X}.

We further extend this result to the Schatten-\(p\) norms, which are a natural extension of \(\ell_p^d\) to matrices. For a \(d\times d\) symmetric real matrix \(X\) and \(p \in [1,\infty]\), the Schatten-\(p\) norm \(\|X\|_{C_p}\) of \(X\) is defined as the \(\ell_p\) norm of the eigenvalues of \(X\). In addition to their intrinsic interest~\cite{LNW14,LW16,LW17,NPS18}, the Schatten-\(p\) spaces are an interesting first step when extending geometric and analytic results from the \(\ell_p\) spaces to more general norms: while Schatten-\(p\) shares many properties with \(\ell_p\), extending proofs and algorithms from \(\ell_p\) to Schatten-\(p\) requires finding coordinate-free arguments that are often more natural and general. Here, we partially succeed in extending Theorem~\ref{thm:ds-ellp} to Schatten-\(p\) spaces, and show the following two theorems. 

\begin{restatable}{theorem}{schattensmall}\label{thm:ds-Sp1}
Let \(\varepsilon \in (0,1]\), \(r > 0\), \(1 \le p \le 2\). For some \(c \lesssim \frac{1}{\varepsilon^{2/p}}\), there exists a data structure for the \((c,r)\)-NNS problem over \(n\)-point sets of \(d\times d\) symmetric matrices with respect to the Schatten-\(p\) norm with 
\begin{itemize}
    \item pre-processing time \(\poly(nd)\);
    \item space \(O(n^{1+\varepsilon}\log(n) \cdot\poly(d))\);
    \item query time \(O(n^{\varepsilon}\log(n) \cdot\poly(d))\).
\end{itemize}
\end{restatable}

The only previously known NNS data structures for Schatten-\(p\) with constant approximation and  \(\poly(nd)\) pre-processing time have polynomial rather than nearly linear space complexity. While not explicitly described there, such data structures follow from the techniques in~\cite{NR06,BG15}, in combination with the results in~\cite{R15}. 

\begin{restatable}{theorem}{schattenlarge}\label{thm:ds-Sp2}
Let \(\varepsilon \in (0,1]\), \(r > 0\), \(p \ge 2\). For some \(c \lesssim \frac{p}{\varepsilon}\), there exists a data structure for the \((c,r)\)-NNS problem over \(n\)-point sets of \(d\times d\) symmetric matrices with respect to the Schatten-\(p\) norm with 
\begin{itemize}
    \item pre-processing time \(\poly(n)\cdot 2^{\poly(d)}\);
    \item space \(O(n^{1+\varepsilon}\log(n) \cdot\poly(d))\);
    \item query time \(O(n^{\varepsilon}\log(n) \cdot\poly(d))\).
\end{itemize}
\end{restatable}

In this theorem, the pre-processing time is exponential in the dimension, similarly to~\cite{spectral}. Nevertheless, the data structure in Theorem~\ref{thm:ds-Sp2} has the benefit that it has query time polynomial in \(d\), rather than polynomial in \(d^p\), as in the data structure in~\cite{spectral}.

\subsection{Metric spaces: notation and assumptions} 

Henceforth, we shall assume that the metric spaces \((\metric,\dist_\metric)\) we deal with are endowed with a dimension \(\dim(\metric)\), which we use to quantify running times of basic tasks, e.g., evaluating distances. We will assume that a point \(x\in\metric\) can be represented by \(\poly(\dim(\metric))\) bits, and that the distance \(\dist_\metric(x,y)\) can be computed in time \(\poly(\dim(\metric))\), as well. We will identify Banach spaces \(X\) with norm \(\|\cdot\|\) with the metric space with distance function \(\dist_X(x,y) = \|x-y\|\). In that case, \(\dim(X)\) is assumed to equal the dimension of $X$ as a vector space. We use \(B_\metric(x,r)\) for the ball of radius \(r\) centered at \(x\) in \(\metric\). For Banach spaces \((X, \|\cdot\|)\), we use \(S_X\) for the unit sphere in \(X\), i.e., for \(\{x \in X: \|x\| = 1\}\).

\subsection{Techniques}\label{subsec:restech}

To prove Theorems~\ref{thm:ds-ellp}, \ref{thm:ds-Sp1}, and~\ref{thm:ds-Sp2}, we develop a new approach for proving partitioning statements such as Theorem~\ref{thm:spectral-ellp} that relies on the notion of embeddings with average distortion, defined below (this definition is taken from~\cite{naor2019average}). We note that average distortion embeddings have been used previously in algorithm design, in particular by Rabinovich for approximation algorithms for the sparsest cut problem~\cite{Rabinovich08}.
\begin{definition}\label{defn:avg-dist}
Given two metric spaces $(\metric, \dist_\metric)$ and $(\calN, \dist_\calN)$, and an $n$-point set $P\subseteq \metric$, we say a function $f: \metric\to \calN$ is an embedding of $\metric$ into $\calN$ with $q$-average distortion $D$ (with respect to $P$) if 
\[
\sum_{x\in P}\sum_{y\in P} \dist_\calN(f(x), f(y))^q \ge \frac{\lip{f}^q}{D^q} 
\sum_{x\in P}^n\sum_{y\in P}^n \dist_\metric(x, y)^q.
\]
\end{definition}
As before, \(\lip{f}\) is the Lipschitz constant of \(f\), i.e., \(\lip{f} = \sup_{x\neq y, x,y\in \metric} \frac{\dist_\calN(f(x), f(y))}{\dist_\metric(x,y)}\). When the embedding has \(1\)-average distortion \(D\), we simply say it has average distortion \(D\). If for every integer \(n\) and any \(n\) point set in \(\metric\), there exists an embedding into \(\calN\) with average distortion \(D\), then we say that \(\metric\) embeds into \(\calN\) with average distortion \(D\).

Briefly, we show that if a metric space \(\metric\) embeds into \(\ell_1^d\) with average distortion \(D\) via an embedding \(f\) that can be efficiently computed from \(P\), and efficiently evaluated, then \(\metric\) supports NNS data structures with approximation \(c \lesssim \frac{D\log D}{\varepsilon}\), space \(O(n^{1+\varepsilon})\),  query time \(O(n^{\varepsilon})\), and efficient pre-processing. The precise statement follows. 
\begin{restatable}{theorem}{avg2NNS}\label{thm:avg2NNS}
Suppose that the metric space \((\metric,\dist_\metric)\) embeds into \(\ell_1^d\) with average distortion \(D\) where \(d \in \poly(\dim(\metric))\). Further, suppose that for any point set \(Q\) of \(m\le n\) points in \(\metric\),  an embedding into \(\ell_1^d\) with average distortion \(D\) with respect to \(Q\) can be computed in time \(T\), and then stored using \(\poly(\dim(\metric))\) bits and evaluated in \(\poly(\dim(\metric))\) time. Then, for any \(\varepsilon \in (0,1]\) and \(\Delta \ge r > 0\), there exists a data structure for the \((c, r)\)-NNS problem over \(n\)-point sets in \(\metric\) of diameter at most \(\Delta\) with 
\begin{itemize}
    \item approximation factor \(c\lesssim \frac{D(1+\log D)}{\varepsilon}\);
    \item pre-processing time \(O(\poly(n\log_{1/b}(n)\dim(\metric))\cdot T)\);
    \item space \(O(n^{1+\varepsilon} \log_{1/b}(n)\poly(\dim(\metric)))\);
    \item query time \(O(n^\varepsilon \log_{1/b}(n)\poly(\dim(\metric)))\);
\end{itemize}
where \(1 - b \gtrsim \frac{r}{\Delta}\).
\end{restatable}


Theorem~\ref{thm:avg2NNS} strengthens the reductions via bi-Lipschitz embeddings mentioned above, and is likely to have applications beyond Theorems~\ref{thm:ds-ellp},~\ref{thm:ds-Sp1},~and~\ref{thm:ds-Sp2}: see, for instance, the general criteria for the existence of embeddings with low average distortion in~\cite{N14a}. The connection between NNS and average distortion embeddings is closely related to the connection between NNS and the \emph{cutting modulus} from prior work~\cite{spectral}. In particular, bounds on the cutting modulus in~\cite{spectral} were proved by utilizing comparison inequalities between non-linear spectral gaps, in the sense of~\cite{N14a}. Such comparison inequalities were shown in~\cite{N14a} to be equivalent to the existence of average distortion embeddings. However, the connection between the cutting modulus and NNS data structures in~\cite{spectral} involves a duality argument that only yields data structures with exponential pre-processing, even when the cutting modulus bound is witnessed by an efficiently computable average distortion embedding. In contrast, Theorem~\ref{thm:avg2NNS} gives an efficient reduction: if the embedding is computationally efficient, so is the data structure, including the pre-processing. 

To prove Theorem~\ref{thm:avg2NNS}, we first formalize the type of data dependent LSH family implicit in Theorem~\ref{thm:spectral-ellp} and in other similar results. As mentioned above, these data-dependent LSH families relax the \(p_2\) requirement of the standard LSH definition by requiring that it holds empirically for the input point set \(P\). I.e., we require that each hash function in the family maps at least \(1-p_2\) fraction of the pairs of points in \(P\) to different values (see Definition~\ref{defn:emp-lsh} for the precise requirement). As noted above, such a data-dependent LSH family is still sufficient to design an NNS data structure with similar running time and space guarantees as given by standard LSH (Lemma~\ref{lm:main-ds}). Moreover, it is also not hard to construct a data-dependent LSH family using a standard LSH family when the point set \(P\) is \emph{dispersed} i.e., when no ball of radius \(cr\) contains more than, e.g., half of the points in \(P\) (Lemma~\ref{lm:lsh2weeklsh}). 

So far these results just give a different perspective on standard NNS data structures using LSH. The benefit of using data-dependent LSH, however, is that the data-dependent requirement allows using a larger class of embeddings in reductions. While the existence of a standard LSH family for \(\ell_1^d\), is inherited by metrics that have a bi-Lipschitz embedding into \(\ell_1^d\) with small distortion,\footnote{In fact a weaker notion of randomized embedding suffices~\cite{AIR18}.} the existence of a data-dependent LSH family for \(\ell_1^d\) is inherited by metrics \(\metric\) that have, for any dispersed point set \(P\subseteq \metric\), an embedding \(f\) into \(\ell_1^d\) which 
\begin{enumerate}[label=(\arabic*)]
    \item  does not expand distances too much, and
    \item does not map a dispersed point set \(P\) in \(\metric\) into a point set \(f(P)\) that is not dispersed in \(\ell_1^d\).
\end{enumerate}
 We formally define this class of embeddings, which we call embeddings with \emph{weak} average distortion, in Definition~\ref{defn:wkavg} below. (``Weak'' here is used in the same sense as weak-\(L_1\) norms.) To complete the connection between NNS and average distortion embeddings, we prove that the existence of (computationally efficient) average distortion embeddings implies the existence of (computationally efficient) weak average distortion embeddings. The proof of this fact uses ideas previously used to relate embeddings with \(q\)- and \(q'\)-average distortion (see Section~5.1 in~\cite{naor2019average}). 

\subsection{Our Results on Average Distortion Embeddings}

Finally, in order to utilize the general connection above between average distortion embeddings and NNS data structures, we need to construct explicit, efficiently computable average distortion embeddings into \(\ell_1^d\) or \(\ell_2^d\). Naor showed that the existence of average distortion embeddings of a metric space \(\metric\) into (infinite dimensional) \(\ell_2\) is equivalent to proving a certain inequality between non-linear spectral gaps, and, using this equivalence, he showed that, when \(p \ge 2\), \(\ell_p^d\) embeds into \(\ell_2\) with 2-average distortion \(D \lesssim p\)~\cite{N14a,naor2019average}. This equivalence between average distortion embeddings and spectral gap inequalities, however, uses a duality argument, and does not provide explicit, efficiently computable embeddings. In fact, an explicit construction of an embedding of \(\ell_p^d\), for \(p \ge 2\), into \(\ell_2\) with \(2\)-average distortion \(O(p)\) is given as an open problem in~\cite{N14a}. Here we resolve this open problem. In the theorem below, the functions \(M_{p,q}, \tilde{M}_{p,q}:\ell_p^d \to \ell_q^d\) are defined by
\begin{align*}
M_{p,q}(x) = \begin{pmatrix}\sign(x_1) |x_1|^{p/q}\\\vdots\\\sign(x_d)|x_d|^{p/q}\end{pmatrix},
&& 
\tilde{M}_{p,q}(x)=
\|x\|_p  M_{p,q}\left(\frac{x}{\|x\|_p}\right) 
= \|x\|_p^{1 - p/q} M_{p,q}(x).
\end{align*}
with \(\tilde{M}_{p,q}(0)=0\).
\begin{restatable}{theorem}{embellp}\label{thm:embedding-ellp}
For any $p \ge q \ge 1$, and any $n$-point set $P$ in $\R^d$, for \(t\in \R^d\) so that 
\begin{equation*}
    \forall i\in [d]:\ \ 
    |\{x\in P: x_i < t_i\}| = |\{x\in P: x_i > t_i\}|,
\end{equation*}
 the map \(g:\ell_p^d \to \ell_q^d\) defined by $g(x) = \tilde{M}_{p,q}(x-t)$ has \(q\)-average distortion \(D \lesssim \frac{p}{q}\). 
\end{restatable}
Theorem~\ref{thm:embedding-ellp}, in the case \(p \ge q =2\), gives an alternative, direct proof of Corollary~1.6 of~\cite{N14a}. 
Naor also points out that, for \(q=2\), the distortion in Theorem~\ref{thm:embedding-ellp} is optimal up to universal constants.

Above, \(M_{p,q}\) is the classical Mazur map from \(\ell_p^d\) to \(\ell_q^d\). This map sends the unit sphere in \(\ell_p^d\) to the unit sphere in \(\ell_q^d\), and, when \(p\ge q\), its restriction to the sphere has Lipschitz constant bounded by \(\frac{p}{q}\) up to absolute constants. The Mazur map was previously used by Matou\v{s}ek to prove bounds on non-linear spectral gaps in \(\ell_p^d\)~\cite{M97}. As mentioned above, such bounds  are closely related to the existence of average distortion embeddings, via Naor's duality argument in~\cite{N14a}. The Mazur map itself, however, cannot be used directly to give an average distortion embedding, since its Lipschitz constant is unbounded over all of \(\ell_p^d\) (see also Remark 1.7 in~\cite{N14a}). Our main technical contribution in the proof of Theorem~\ref{thm:embedding-ellp} is the realization that the re-scaled Mazur map \(\tilde{M}_{p,q}\) has Lipschitz constant \(\lesssim \frac{p}{q}\) \emph{everywhere}, and that the machinery of Matou\v{s}ek's argument can be then used to prove a bound on the average distortion directly, without going through a duality argument. 

Our technique for constructing explicit average distortion embeddings in fact extends to every pair of normed spaces \(X\) and \(Y\) for which we have a H\"older-continuous homeomorphism \(f\) between the unit spheres of \(X\) and \(Y\). We can then show that this homeomorphism can be extended to a function \(\tilde{f}\) which is H\"older-continuous on all of \(X\), and that there is a shift \(t\in X\) so that the map \(g:X \to Y\) defined by \(g(x) = \tilde{f}(x-t)\) is an average distortion embedding of (a snowflake of) \(X\) into \(Y\). One can then use a variety of known homemorphisms between spheres and construct reasonably explicit average distortion embeddings. We do so for the Schatten-\(p\) spaces, and one can also use the homeomorphism between finite dimensional normed spaces in~\cite{daher} to give results for general normed spaces, although we do not pursue this here. Except for some special cases like \(\ell_p^d\) and Schatten-\(p\) for \(1 \le p \le 2\), however, one aspect of these embeddings is still not fully explicit and in particular, not computationally efficient. Namely, the argument showing that there exists a good shift \(t\), which was first given in~\cite{spectral}, uses the theory of topological degree that is also used in textbook proofs of Brouwer's fixed point theorem, and does not suggest an efficient algorithm for computing \(t\). We leave finding such an algorithm, even for the case of Schatten-\(p\) norms with \(p\ge 2\), as an open problem.

\section{Data-Dependent LSH Families}

In this section, we introduce a formalization of the data-dependent LSH families we are going to use. We show how to use such LSH families to construct a data structure for the NNS problem, by generalizing the randomized decision tree data structure from~\cite{spectral}.

\subsection{Definitions and Basic Facts} 

The following definition is now standard and goes back to Indyk and Motwani~\cite{IM98}.
\begin{definition}[LSH Family]
Let \((\metric,\dist_\metric)\) be a metric space and fix a scale \(r > 0\), approximation factor \(c > 1\), and range \(\range\). Then a probability distribution \(\lsh\) over maps from \(\metric\) to \(\range\) is called \((r,cr,p_1,p_2)\)-sensitive if 
\begin{align*}
    \dist_\metric(x,y) \le r &\implies
    \Pr_{h\sim\lsh}[h(x) = h(y)] \ge p_1,\\
    \dist_\metric(x,y) > cr &\implies
    \Pr_{h\sim\lsh}[h(x) = h(y)] < p_2.\\
\end{align*}
\end{definition}

Our data structure is based on the following, in a sense, weaker definition which allows \(\lsh\) to depend on the point set, and defines \(p_2\) in terms of the how hash functions spread the points among the bins they are hashed to.

\begin{definition}\label{defn:emp-lsh}
Let \((\metric,\dist_\metric)\) be a metric space and fix a scale \(r > 0\), and range \(\range\). Let \(P\subseteq \metric\) be an $n$-point set. A probability distribution \(\lsh\) over maps from \(\metric\) to \(\range\) is called \((r, p_1, p_2)\)-empirically sensitive for \(P\) if 
\begin{align*}
    \dist_\metric(x,y) \le r \implies
    &\Pr_{h\sim\lsh}[h(x) = h(y)] \ge p_1,\\
 \forall \bin \in \range, \forall h \in \supp(\lsh):
    \ \ \ &|\{x\in P: h(x) = \bin\}| \le p_2 n.
\end{align*}
\end{definition}

We call families of hash functions \(\lsh(P)\) as in Definition~\ref{defn:emp-lsh} \emph{data-dependent locality sensitive hash functions}, because \(p_2\) restricts the empirical distribution of the points in \(P\) among the bins defined by the hash function. This definition was implicit in prior work, e.g.~\cite{spectral}.

The following definition captures the property of a point set being well-spread on a given scale. We will need this property in order to be able to construct a data-dependent LSH family with good parameters.
\begin{definition}
Let \((\metric,\dist_\metric)\) be a metric space and let \(t > 0\).
A set \(P\) of \(n\) points in \(\metric\) is called \((t,\disp)\)-dispersed if, for all \(x \in \metric\), \(|P \cap B_\metric(x,t)| \le (1-\disp) n\).
\end{definition}

We could have also defined the dispersed point sets with respect to ball centers \(x\) in the point set \(P\). 
The following simple lemma shows that the two definitions are essentially equivalent.
\begin{lemma}\label{lm:dispersed-P}
Suppose that, for a set \(P\) of \(n\) points in a metric space \(\metric\), for every $x_0\in P$, we have \(|P \cap B_\metric(x_0,2t)| \le (1-\disp) n\). Then $P$ is \((t,\disp)\)-dispersed.
\end{lemma}
\begin{proof}
Suppose, for contradiction, there were some \(x\in \metric\) such that \(|P \cap B_\metric(x,t)| > (1-\disp) n\). Then, for any \(x_0\in P \cap B_\metric(x,t)\), we  have \(B_\metric(x,t) \subseteq B_\metric(x_0,2t)\) by the triangle inequality, and \(|P \cap B_\metric(x_0,2t)| > (1-\disp) n\), contradicting our assumption. 
\end{proof}

For a metric space \((\metric, \dist_\metric)\), and a (multi-)set \(P\subseteq \metric\) of size \(n\), we use the notation
\[
\Psi_\metric(P,t)
= 
\frac{|\{(x,y) \in P \times P: \dist_\metric(x,y) > t\}|}{n^2}.
\]

The following lemma relates the notion of being \((r,\beta)\)-dispersed and the function \(\Psi_\metric(P,t)\). 
\begin{lemma}\label{lm:disp-cdf}
Let \((\metric, \dist_\metric)\) be a metric space, and let \(P\) be a set of \(n\) points in \(\metric\). If \(P\) is \((t,\beta)\)-dispersed, then
\(
\Psi_\metric(P,t) \ge \beta. 
\)
Conversely, \(P\) is \((t, \beta)\)-dispersed for \(\beta = \frac12\Psi_\metric(P,2t)\).
\end{lemma}
\begin{proof}
If \(P\) is \((t,\beta)\)-dispersed, then, for any \(x\in P\),
\[
\frac{|\{y \in P \times P: \dist_\metric(x,y) > t\}|}{n} \ge \beta.
\]  
Averaging these inequalities over all \(x\in P\) proves that \(\Psi_\metric(P,t)\ge \beta\).

In the other direction, let \(\beta = \frac12\Psi_\metric(P,2t)\), and suppose, towards contradiction, that \(P\) is not \((t,\beta)\)-dispersed. Then, there exists some \(x_0\in \metric\) such that \(|P \cap B_\metric(x_0,t)| > (1-\beta) n\). By the triangle inequality,
\[
\frac{|\{(x,y) \in P \times P: \dist_\metric(x,y) \le 2t\}|}{n^2}
\ge 
\frac{|P \cap B_\metric(x_0,t)|^2}{n^2}
> (1-\beta)^2.
\]
Therefore, \(\Psi_\metric(P,2r) < 1 - (1-\beta)^2 < 2\beta,\) a contradiction.
\end{proof}

\subsection{NNS Data Structures from Data-Dependent LSH Families}

Next we show that a data-dependent LSH family can be used to construct an efficient NNS data structure. As a basic building block, we use the classical linear space, constant query time static dictionary data structure, whose existence is guaranteed by the following lemma.

\begin{lemma}[Efficient Static Dictionary \cite{FredmanKS84}]\label{lem:perfect-hash}
Let $S$ be a set of $s$ keys from a universe $U$. There exists a data structure 
that can be constructed in $\poly(s,\log |U|)$ time, that requires $O(s\log |U|)$ bits of space to store, and supports a query algorithm that, given input $u\in U$, determines whether $u\in S$ in $O(1)$ time. The algorithm outputs $\perp$ if $u\notin S$.
\end{lemma}

The next lemma shows that data-dependent LSH families imply efficient NNS data structures.

\begin{lemma}\label{lm:main-ds}
Let \((\metric,\dist_\metric)\) be a metric space and let \(r > 0\), and \(c > 1\). Suppose that for every \((cr,\frac12)\)-dispersed \(m\)-point set \(Q \subseteq \metric\), there exists a \((r,p_1(Q),p_2(Q))\)-empirically sensitive \(\lsh(Q)\) such that \(\frac{\log(1/p_1(Q))}{\log(1/p_2(Q))} \le \rho\) where all \( p_2(Q) \le p_2\) for some $p_2\in (0,1)$. Define $b = \max(\frac12,p_2)$. Further, suppose that any \(h\) in the support of \(\lsh(Q)\) can be stored in space and evaluated in time polynomial in \(\dim(\metric)\), that \(h\sim \lsh(Q)\) can be sampled in time \(T_s(m)\), for a non-decreasing function \(T_s(m)\), and the range $\range$ of $\lsh(Q)$ has size at most $\exp(\poly(\dim(\metric)))$. Then there exists a data structure for the \((2c+1,r)\)-NNS  problem over \(n\)-point sets in \(\metric\) with 
\begin{itemize}
    \item pre-processing time \(O(\poly(n\log_{1/b}(n)\dim(\metric))\cdot T_s(n))\);
    \item space \(O(n^{1+\rho}\log_{1/b} (n) \cdot\poly(\dim(\metric)))\);
    \item query time \(O(n^\rho \log_{1/b}(n)\cdot  \poly(\dim(\metric)))\).
\end{itemize}
\end{lemma}

The proof and terminology are largely adapted from \cite{spectral}. The data structure is a randomized decision tree similar to the one in~\cite{spectral}, with two generalizations: we allow the hash function to split space into more than two parts, and more importantly, we allow the parameters of the data-dependent LSH family to change from one node of the decision tree to the next.

\begin{proof}
Let $d = \dim(\metric)$ and $P\subseteq \metric$ be a set of $n$ points. The data structure consists of $k$ many independently generated
random decision trees. Each node $v$ of a tree stores the following attributes:

\begin{itemize}
    \item $v.$type: the type of the node;
    \item $v.P$: a subset of the dataset points;
    \item $v.$center:  a point in $\metric$;
    \item $v.h$: the $\poly(\dim(\metric))$ bits needed to store a hash function $h$ from a family $\lsh(Q)$;
    \item $v.$dict: the $O(|Q|\cdot d)$ bits needed to store the static dictionary from Lemma \ref{lem:perfect-hash} with $h(Q)$ as its set of keys; together with each element of \(h(Q)\) we store a pointer to a corresponding child of \(v\).
\end{itemize}

\begin{figure}[h]
\begin{subfigure}[t]{0.5\textwidth}
\centering
{\footnotesize
\begin{algorithmic}
\Function{Process}{$Q$, $\ell$, $v$}
\If{$\ell = t$ or $|Q| \leq 100$}
\State $v$.type $\leftarrow$ ``leaf''
\State $v.P \leftarrow Q$
\ElsIf{$\exists x_0\in Q$ s.t.~$|Q\cap B_\metric(x_0, 2cr)| > \frac{|Q|}{2}$}
\State call \textsc{ProcessBall}$(Q, x_0, \ell, v)$
\Else
\State Sample $h\sim \lsh(Q)$ and store it in $v.h \leftarrow h$
\State \textsc{ProcessHash}$(Q, h, \ell, v)$
\EndIf
\EndFunction
\\

\end{algorithmic} }
\end{subfigure}
\qquad
\begin{subfigure}[t]{0.5\textwidth}
\centering
{\footnotesize
\begin{algorithmic}
\Function{ProcessBall}{$Q$, $x$, $\ell$, $v$}
\State $v.$type $\leftarrow$ ``ball''
\State $v.$center $\leftarrow x$
\State create a node $v.$child 
\State call $\textsc{Process}(Q\setminus B_\metric(x,2cr),\ell+1, v.$child)
\EndFunction
\\ 
\Function{ProcessHash}{$Q, h, \ell, v$}
\State $v$.type $\leftarrow$ ``hash''
\State evaluate the set $h(Q)= \{\bin_1,\ldots ,\bin_s\}$
\State construct the static dictionary from Lemma \ref{lem:perfect-hash} with $h(Q)$ \quad as its set of keys and store it in $v.$dict
\For{$j = 1$ to $s$}
\State create a node $v.\bin_j$ 
\State call $\textsc{Process}(h^{-1}(\bin_j), \ell + 1, v.\bin_j)$
\EndFor
\EndFunction

\end{algorithmic} }
\end{subfigure}
\caption{Pseudocode for constructing the data-structure}
\label{pseudo-build}
\end{figure}

\paragraph{Pre-processing:} We keep a counter $\ell$ which denotes the current level of the tree we are processing. Initially, $\ell = 0$,
and it is incremented on each recursive call. Once $\ell$ reaches the threshold $t\coloneqq \lceil \log_{1/b}(n)\rceil$,
we store a leaf node $v$ and save the points of the dataset which reached $v$ in $v.P$. Thus the depth of
the tree is bounded by $t$ \emph{a priori}. For a non-leaf node $v$ which is processed by a call to $\textsc{Process}(Q,\ell,v)$, we do one of the following.

\begin{itemize}
    \item If there exists a point $x_0 \in Q$ such that $|Q \cap B_\metric(x_0, 2cr)| > |Q|/2$ , we build a \emph{ball} node. In this
case, the ball node saves $x_0$ in $v.$center. 
We then recurse by building
a data structure on $Q \setminus B_\metric(x_0, 2cr)$.
\item Otherwise, by Lemma \ref{lm:dispersed-P}, $Q$ is $(cr,\frac{1}{2})$-dispersed and by assumption, we have a \((r,p_1(Q),p_2(Q))\)-empirically sensitive LSH family \(\lsh(Q)\). We sample a hash function $h\sim \lsh(Q)$ and build a \emph{hash} node $v$. We store the $\poly(\dim(\metric))$ bits necessary to store a hash function from the family $\lsh(Q)$ in $v.h$, and recursively create a child node corresponding to each element in $h(Q)$. 

\end{itemize}

The final data structure consists of $k = O(n^\rho)$ independent trees, rooted at the nodes $v_1,\ldots, v_k$, where the $i^\mathrm{th}$ tree is built by a call to $\textsc{Process}(P, 0, v_i)$.

\begin{figure}[h]
\begin{subfigure}[t]{0.5\textwidth}
\centering
{\footnotesize
\begin{algorithmic}
\Function{Query}{$q$, $v$}
\If{$v.\type = \mathrm{``leaf"}$}
\For{$x$ in $v.P$}
\State \Return $x$ if $\dist_\metric(x,q) \leq (2c+1)r$
\EndFor
\State \Return $\perp$
\ElsIf{$v.\type = \mathrm{``ball"}$}
\State $x\leftarrow \textsc{QueryBall}(q,v)$
\State \Return $x$ if $x\neq \perp$
\ElsIf{$v.\type = \mathrm{``hash"}$}
\State $x\leftarrow \textsc{QueryHash}(q,v)$
\State \Return $x$ if $x\neq \perp$
\EndIf
\EndFunction
\\

\end{algorithmic} }
\end{subfigure}
\qquad
\begin{subfigure}[t]{0.5\textwidth}
\centering
{\footnotesize
\begin{algorithmic}
\Function{QueryBall}{$q, v$}
\State $x_0\leftarrow v.\cent$
\If{$\dist_\metric(x_0,q)\leq (2c+1)r$} 
\State \Return $x_0$
\EndIf
\State \Return $\textsc{Query}(q,v.\mathrm{child})$

\EndFunction
\\ 
\Function{QueryHash}{$q, v$}
\State identify $h$ from $v.h$
\State check if $h(q)$ is in the set of keys of $v$.dict
\If{the output of the algorithm is \emph{not} $\perp$} 
\State \Return $\textsc{Query}(q,v.h(q))$
\Else\ \Return $\perp$
\EndIf
\EndFunction
\end{algorithmic} }
\end{subfigure}
\caption{Pseudocode for querying the data-structure}
\label{pseudo-query}
\end{figure}

\paragraph{Querying the Data Structure:}
We now specify how to query the data structure; the pseudocode
is given above. For each of the $k$ trees in the data structure, we start the query procedure at
the root of the tree, and proceed by cases according to the type of node, as follows:

\begin{itemize}
    \item \emph{Leaf nodes:} If a query $q\in\metric$ queries a leaf node $v$, then we scan $v.P$ and return the
first point which lies within distance $(2c+1)r$. If no such point is found, we return $\perp$.

\item \emph{Ball nodes:} If a query $q\in\metric$ queries a ball node $v$, we test whether the query is close to the
ball centered at $x_0 = v.$center of radius $2cr$. In particular, if $\dist_\metric(x_0, q) \leq (2c+1)r$,
we return $x$. Otherwise, we recurse on the (unique) child node of $v$.

\item \emph{Hash nodes:} If a query $q\in\metric$ queries a hash node $v$, the querying algorithm checks if $q$ lies in the same ``cell'' as one of the dataset points by evaluating $h(q)$ (using the hash function stored in $v.h$) and seeing if it is in $h(P)$. This last check is achieved efficiently via the use of the efficient static dictionary data structure from Lemma \ref{lem:perfect-hash}. If it does not lie in one of these cells, the algorithm returns $\perp$. Otherwise, we descend down the child of $v$ in the tree corresponding to the cell $h(q)$ (a pointer to this child is stored in the dictionary) and recurse. 
\end{itemize}

We collect some simple observations about the data structure below.

\begin{claim}\label{clm:obs}
The following statements hold:
\begin{itemize}
    \item For a hash node $u$ that is built by a call to $\textsc{Process}(P_u,\ell, u)$ for some $P_u \subseteq P$ and has children $u_1,\ldots,u_s$ built by calls to $\textsc{Process}(P_{u_i},\ell, u_i)$ respectively, the sets $\{P_{u_i}\}$ form a partition for $P_u$.
    \label{subclm:unique}
    \item If $\textsc{Query}(q,v)$ returns a point $x$, then $\dist_\metric(q,x) \leq (2c+1) r$.
\end{itemize}
\end{claim}

\paragraph{Correctness:}
We summarize the proof of correctness in the following claim.

\begin{claim}
If, for a query point $q\in \metric$, there is a point $x^*\in P$ in the dataset such that $\dist_\metric(x^*,q)\leq r$, then with probability at least $2/3$, our algorithm outputs $x\in P$ with $\dist_\metric(x,q)\leq (2c+1)r$.
\end{claim}
\begin{proof}
Consider the fixed dataset $P$ and query $q\in \metric$ which is promised to have a point
$x^* \in P$ with $\dist_\metric(x^*, q) \leq r$. 
By Claim \ref{clm:obs}, any point $x$ returned by $\textsc{Query}(q, v_i)$, where $v_i$
is the root of one of the $k$ data
structure decision trees, indeed satisfies $\dist_\metric(x, q) \leq (2c+1) r$. To prove correctness, it remains to argue that, with
sufficiently high probability, at least one of the $\textsc{Query}(q, v_i)$ calls, for $i = 1,\ldots, k$, does in fact
return a point. Fix any $i$ between $1$ and $k$, and consider the (random) \emph{branch} of that decision tree traversed by the point $x^*$. More specifically, there is a random sequence of nodes $U_0, \ldots, U_\ell$ such that one of the following two cases must occur:

\begin{enumerate}[label=(\roman*)]
    \item\label{it:leaf} $U_0 = v_i$; $U_{j+1}$ is a child of $U_j$ for $j = 0,\ldots,\ell-1$; $U_\ell$ is a leaf node with $x^*\in U_\ell.P$, or
    \item\label{it:ball} $U_0 = v_i$; $U_{j+1}$ is a child of $U_j$ for $j = 0,\ldots,\ell-1$; $U_\ell$ is a ball node with $\dist_\metric(U_\ell.\text{center},x^*)\le 2cr$.
\end{enumerate}

 I.e., the {branch} of the decision tree traversed by $x^*$``ends'' either at a leaf node or at a ball node. Now, let $\calS$ denote the \emph{success} event i.e., the event that the call $\textsc{Query}(q, v_i)$ returns a point. Further, for a node $v$ in the tree rooted at $v_i$, let $\calE_{v}$ denote the event that $\textsc{Query}(q, v)$ is called during the execution of the querying algorithm. Let also $P_v$ be the subset of the dataset $P$ involved in the call
$\textsc{Process}(P_v ,j, v)$ that created $v$ (where $j$ is the distance of $v$ from the root $v_i$). We show, via a backward induction on $j$, that in either of the two cases above, we have with probability $1$, that \[\Pr[\calS\ |\ \calE_{U_j}] \geq |P_{U_j}|^{-\rho}.\] 
For $j=0$, this would imply that $\Pr[\calS\ |\ \calE_{U_0}] = \Pr[\calS\ |\ \calE_{v_i}] = \Pr[\calS] \geq n^{-\rho}$. Thus, by
choosing the number $k$ of trees in the data structure to be a sufficiently large multiple of $n^\rho$, we
can guarantee that with probability at least $2/3$, the data structure indeed returns a point $x$ such that
$\dist_\metric(x, q) \leq (2c+1)r$.

For the base case $j = \ell$, it is easy to see that $\Pr[\calS\ |\ \calE_{U_\ell}] = 1$: if $U_\ell$ is a leaf node (Case \ref{it:leaf}), then since the algorithm searches among all points $x$ in $U_\ell.P$, it must encounter a point $x$ with $\dist_\metric(x,q)\leq (2c+1)r$ because $x^*$ itself is a such a point. And if $U_\ell$ is a ball node (Case \ref{it:ball}), then the algorithm returns the center $x$ of $U_\ell$, for which $\dist_\metric(x,q) \leq \dist_\metric(x,x^*) + \dist_\metric(x^*,q) \leq (2c+1)r$.

Next, suppose that $\Pr[\calS\ |\ \calE_{U_j}] \geq |P_{U_j}|^{-\rho}$ for some $j\in [\ell]$. We now prove the hypothesis for $j-1$.
{
First suppose that $U_{j-1}$ is a ball node with center $x$. There are two cases: either $\dist_\metric(x,q)\leq (2c+1)r$ or $\dist_\metric(x,q)> (2c+1)r$. In the former case, the querying algorithm, by design, returns the center $x$ and so given $\calE_{U_{j-1}}$, we are guaranteed success with probability $1$. In the latter case, the algorithm necessarily calls $\textsc{Query}(q, U_j)$ and so, we have that $\calE_{U_j}$ holds if $\calE_{U_{j-1}}$ does, implying
\[
\Pr[\calS\ |\ \calE_{U_{j-1}}] = \Pr[\calS\ |\ \calE_{U_{j}}] \geq |P_{U_j}|^{-\rho}> |P_{U_{j-1}}|^{-\rho} 
\]
as $|P_{U_j}|< |P_{U_{j-1}}|/2$ since $U_{j-1}$ is a ball node.}

Next, suppose that $U_{j-1}$ is a hash node. First, notice that we have
\begin{equation}\label{eq:satisfy}
\Pr_{h\sim\lsh(P_{U_{j-1}})}[\calE_{U_j}\ |\ \calE_{U_{j-1}}] \geq \Pr_{h\sim\lsh(P_{U_{j-1}})}[h(q)=h(x^*)\ |\ \calE_{U_{j-1}}]\geq p_1(P_{U_{j-1}})
\end{equation}
from the empirical sensitivity of $\lsh(P_{U_{j-1}})$.
Also, note that the event $\calE_{U_{j}}$  can hold only if $\calE_{U_{j-1}}$ does, i.e., $\calE_{U_{j-1}}\supseteq \calE_{U_{j}}$ and so,
\[
\Pr[\calS\ |\ \calE_{U_{j-1}}] = \frac{\Pr[\calS \cap \calE_{U_{j-1}}]}{\Pr[\calE_{U_{j-1}}]} \geq \frac{\Pr[\calS \cap \calE_{U_{j}}]}{\Pr[\calE_{U_{j}}]}\cdot \frac{\Pr[\calE_{U_{j}}]}{\Pr[\calE_{U_{j-1}}]} \geq |P_{U_j}|^{-\rho}\cdot p_1(P_{U_{j-1}})
\]
where the final inequality follows from (\ref{eq:satisfy}) and the induction hypothesis. Next, as {$|P_{U_j}| \leq p_2(P_{U_{j-1}})\cdot |P_{U_{j-1}}|$}, it follows that 
\[
|P_{U_j}|^{-\rho} \cdot p_1(P_{U_{j-1}})\geq (p_2(P_{U_{j-1}})\cdot|P_{U_{j-1}}|)^{-\rho}\cdot p_1(P_{U_{j-1}}) \geq |P_{U_{j-1}}|^{-\rho}
\]
where the final inequality follows from the given assumption that \(\frac{\log(1/p_1(P_{U_{j-1}}))}{\log(1/p_2(P_{U_{j-1}}))} \le \rho\).
\end{proof}

\paragraph{Pre-processing Time Analysis:} 
Let us bound the time it takes to execute $\textsc{Process}(P, 0, v_1)$ i.e., to construct \emph{one} of the $k$ decision trees. 
We claim that the total time required to process all the nodes at any level $\ell$ is \(O(\poly(n\dim(\metric))\cdot T_s(n))\).

First note that there are at most $n$ nodes at each level: this is because any point $p\in P$ is only processed at at most one node on any given level. Since there are at most $t = O(\log_{1/b} n)$ many levels, it suffices to bound the time it takes to process a node \emph{before} recursing on one of its children. 

So suppose that  at some point the algorithm calls $\textsc{Process}(Q, \ell, v)$. It then decides to further process this node as a leaf, ball or hash node. Processing it a leaf node simply takes $O(|Q|)$ time. Checking for a dense ball, i.e., checking if there is an $x_0\in Q$ such that $|Q\cap B_\metric(x_0,2cr)| > \frac{|Q|}{2}$ takes at most $O(|Q|^2)$ time. Processing it as a ball node with arguments $(Q,x,\ell,v)$ takes $O(|Q|)$ time (before descending down its child).  Finally if the algorithm decides to process it as a hash node, sampling an $h\sim \lsh(Q)$ takes $T_s(|Q|)$ time and then, for processing a hash node with arguments $(Q,h,\ell,v)$, first note that the evaluation of the set $h(Q)$ takes $|Q|\cdot \poly(\dim(\metric))$ time; further, $|h(Q)|\leq |Q|$ and so  construction of the static dictionary takes $\poly(|Q|\dim(\metric))$ time.

In total, we conclude that the amount of time spent in a node is at most \(O(\poly(n\dim(\metric))\cdot T_s(n))\), regardless of its type. The final bound on the pre-processing time follows from the observation that there are at most $O(nt)$ many nodes in one tree, and at most $k = o(n)$ many trees.

\paragraph{Space Analysis:}
Note that a node $v$ stores $O(|P_v|\cdot \poly(\dim(\metric)))$ bits of memory in all for all its attributes. 
Since each point in $P$ ``contributes" to exactly one node per level of a tree, this amounts to storing at most $O(n\cdot\poly(\dim(\metric)))$ many bits per level. Taking into account that there are $O(\log_{1/b}(n))$ levels per tree and at most $O(n^\rho)$ many trees, we obtain the desired space bound. 

\paragraph{Query Time Analysis:}
In each of the $k = O(n^\rho)$ decision trees, the querying procedure on a given input point $q$ simply specifies a branch from the root of that decision tree to a leaf. The length of this branch is bounded by the height of the decision tree which is $O(\log_{1/b}(n))$. The amount of time spent in an internal node is always at most $\poly(\dim(\metric))$: it is constant time for a ball node (to check if the distance of $q$ from the center is at most $(2c+1)r$) and at most $\poly(\dim(\metric))$ in the case of hash nodes, in order to first evaluate $h(q)$ and then perform a lookup in the dictionary. Therefore, it suffices to verify that the amount of time spent in a leaf node is constant; in other words, that $|v.P|$ is constant for a leaf node $v$. 

For any leaf nodes at level $\ell < t$, this is immediate by construction. For any leaf at level $t$ of a decision tree, consider the unique branch from the root to that leaf and two nodes $u$ and its child $v$ along that branch. Suppose that $u$ was created by a call to $\textsc{Process}(P_u, \ell, u)$ and $v$ by $\textsc{Process}(P_v, \ell + 1, v)$. Either $u$ is a ball node in which case we know that $|P_v| < |P_u|/2$, or $u$ is a hash node with $u.h = h$ in which case we know that $P_v$ is a ``cell'' of $h(P_u)$ i.e., $P_v = h^{-1}(\bin)$ for some $\bin\in h(P_u)$. But by the empirical sensitivity guarantee for $h\sim \lsh(P_u)$, we know that $|h^{-1}(\bin)| \leq p_2(P_u)\cdot |P_u| \leq p_2 \cdot |P_u|$. Therefore, the number of points that can reach a leaf of any tree is bounded by $n\cdot b^t$ which is bounded by a constant due to the choice of $t$.

\end{proof}

\subsection{Data-dependent LSH Families from LSH Families}

We show that the existence of an LSH family implies the existence of a data-dependent LSH family. 

\begin{lemma}\label{lm:lsh2weeklsh}
Let \((\metric,\dist_\metric)\) be a metric space and let \(r > 0\), and \(c > 1\). Suppose there exists a \((r,cr,p_1,p_2)\)-sensitive \(\lsh\) where $p_2 \le 1/2$. Let \(P\subseteq \metric\) be a \((cr,\beta)\)-dispersed \(n\)-point set. Then, for 
\[
p'_2 = \sqrt{1 - \beta(1-2p_2)} 
\]
the event \(\calE = \{\max_{\bin\in\range} |\{x\in P: h(x) = \bin\}| \le p'_2 n\}\) occurs with probability at least \(1/2\), and \(\lsh\) conditioned on \(\calE\) is \((r,2p_1-1,p'_2)\)-empirically sensitive for \(P\).
\end{lemma}

\begin{proof}
Let \(m = \Psi_\metric(P,cr)n^2\) be the number of ordered pairs of far points, and note that, by Lemma~\ref{lm:disp-cdf} and the assumption that \(P\) is \((cr,\beta)\)-dispersed, \(m \ge \beta n^2\). Further, let $X$ denote the random variable  (over the random choice of $h\sim \lsh$) that counts the number of far ordered pairs that collide i.e., 
\[
X = |\{(x,y) \in P \times P: \dist_\metric(x,y) > cr, h(x) = h(y)\}|.
\]
Let $\chi(\cdot)$ denote the indicator random variable for an event. Notice that
\[
X = \sum_{\substack{(x,y)\in P\times P :\\ \dist_\metric(x,y)>cr }} \chi(h(x) = h(y)) \leq m
\]
and therefore, $\E_{h\sim\lsh}[X] < p_2 m$ by the sensitivity of $\lsh$.
Now, let $Y = |\{(x,y)\in P \times P: h(x) = h(y)\}|$ be the random variable that denotes the total number of colliding ordered pairs of data points. We have
\begin{align*}
Y &= \sum_{x\in P}\sum_{y\in P} \chi((h(x) = h(y))\\
    &=\sum_{\substack{(x,y)\in P\times P :\\ \dist_\metric(x,y)\le cr}} \chi(h(x) = h(y))+ \sum_{\substack{(x,y)\in P\times P :\\ \dist_\metric(x,y)> cr}} \chi(h(x) = h(y))\\
    &\le n^2 - m + X 
\end{align*} 
By Markov's inequality, with probability at least $1/2$, $X < 2 p_2 m$, and, therefore, \(Y \le n^2 - (1-2p_2)m\). Using the assumption \(p_2 \le \frac12\) and the inequality $m \ge \beta n^2$, this gives us
\[
\Pr[Y > n^2(1-\beta(1-2p_2))] < \frac{1}{2}
\]
or in other words, that $Y > {p_2'}^2n^2$ with probability at most $1/2$. To see the proof through, suppose that there exists $\bin\in\range$ with $|\{x\in P: h(x) = \bin\}| > p'_2 n$. Then, every point in this \emph{cell} corresponding to $\bin$ collides with every other point in the cell, and we have $Y\geq |\{x\in P: h(x) = \bin\}|^2 > {p'_2}^2 n^2$. But we have already shown that the probability of this event is at most $1/2$. Thus, the event $\calE$ occurs with probability at least $1/2$.

Finally, we need to show that conditioned on $\calE$, $\lsh$ is \((r,2p_1-1,p'_2)\)-empirically sensitive for \(P\). For points $x,y\in P$ such that $\dist_\metric(x,y) \leq r$, note that
\[
\Pr_{h\sim \lsh}[h(x) \neq h(y)\ |\ \calE] = \frac{\Pr[\{h(x) \neq h(y)\}\wedge \calE] }{\Pr[\calE]} \le \frac{1-p_1}{1/2} = 2-2p_1
\]
and the conclusion follows.
\end{proof}

We instantiate this lemma with the LSH for \(\ell_1^d\) due to Indyk and Motwani~\cite{IM98} (see also~\cite{GIM99}). 
\begin{lemma}\label{lm:l1-lsh}
For any \(r > 0\), any \(c>1\), and any \(\Delta \ge 1\), the space \(\ell_1^d\) restricted to \([-\Delta, \Delta]^d\) has a \((r,cr,p_1,p_2)\)-sensitive \(\lsh\) with \(p_1 = 1 - \frac{r}{2d\Delta}\) and \(p_2 = 1 - \frac{cr}{2d\Delta}\). Moreover, \(h\sim \lsh\) can be sampled and evaluated in constant time.
\end{lemma}

Together with Lemma~\ref{lm:lsh2weeklsh}, Lemma~\ref{lm:l1-lsh} implies the following corollary. 
\begin{corollary}\label{cor:ell1-wlsh}
Let \(r > 0\), \(c > 6\), \(\Delta \ge 1\), and let \(P\) be a \((cr, \beta)\)-dispersed \(n\)-point set in \(\ell_1^d\) restricted to \([-\Delta, \Delta]^d\). There exists a \((r,1 - \frac{8}{c},1-\frac{\beta}{4})\)-empirically sensitive \(\lsh_{\ell_1^d}(P)\) for \(P\). Moreover, a function \(h \sim \lsh(P)\) can be sampled in \(O(n\ \poly(d\Delta/r))\) time, evaluated in \(\poly(d\Delta/r)\) time, and stored using \(\poly(d\Delta/r)\) bits.
\end{corollary}
\begin{proof}
Let \(\lsh\) be the \((r,cr,p_1,p_2)\)-sensitive hash family from Lemma~\ref{lm:l1-lsh}. We will assume that \(\Delta \ge \frac{cr}{2d}\), since otherwise we can increase \(\Delta\) without affecting the asymptotic complexity of the hash functions. For any integer \(k\ge 1\) we can define a  \((r,cr,p_1^k,p_2^k)\)-sensitive \(\lsh^k\)  by taking \(k\) independent samples \(h_1, \ldots, h_k\) from \(\lsh\) and defining \(h \sim \lsh^k\) as \(h(x) =(h_1(x), \ldots, h_k(x))\). Let us then choose 
\(k\) to be the smallest integer 
so that \(p_2^k \le \frac14\). Then, for \(\rho = \frac{\log(1/p_1)}{\log(1/p_2)}\),
\(
p_1^k = p_2^{-(k-1)\rho}p_1 \ge 4^{-\rho} p_1.
\)
Using the inequality \(\rho \le \frac1c - \frac1c\ln(p_1)\), this gives us 
\[
p_1^k > 2^{-2/c} p_1^2 > \left(1-\frac{2}{c}\right)\left(1-\frac{r}{2d\Delta}\right)^2
> 1- \frac{4}{c},
\]
where the last inequality uses the assumption \(\Delta \ge \frac{cr}{2d}\).
The corollary now follows from Lemma~\ref{lm:lsh2weeklsh}.
\end{proof}

\section{Weak Average Distortion Embeddings}

Next, we  introduce the notion of weak average distortion embedding, and show that weak average distortion embeddings into \(\ell_1^d\) or \(\ell_2^d\) imply the existence of data-dependent LSH. We start with the definition.

\begin{definition}\label{defn:wkavg}
Let \((\metric, \dist_\metric)\) and \((\mathcal{N}, \dist_{\mathcal{N}})\) be  two metric spaces, and let \(P\subseteq \metric\) be an \(n\)-point set. 
A function \(f:\metric\to \mathcal{N}\) is an embedding with weak average distortion \(D\) with respect to \(P\) if
we have
\[
\sup_{t \ge 0} t\Psi_{\mathcal{N}}(f(P),t)
\ge \frac{\lip{f}}{D} \sup_{t \ge 0} t\Psi_{\metric}(P,t). 
\]
\end{definition}
The name ``weak average distortion'' comes from the fact that \(\sup_{t \ge 0} t\Psi_{\metric}(P,t)\) is the weak-$L_1$ norm of \(\dist_\metric\) with respect to the uniform measure over \(P \times P\). So, a weak average distortion embedding is required to not expand distances too much while also not decreasing the weak-\(L_1\) norm of the pairwise distances. The analogous notion of $q$-average distortion, where we instead take the $L_q$ norm of \(\dist_\metric\) with respect to the same measure (see Definition~\ref{defn:avg-dist}), has been studied before. The definition of weak average distortion embedding appears to be new. It can be extended in a natural way to more general probability measures, but we will not pursue this here.

In this subsection, we show that a weak average distortion embedding of \(\metric\) into \(\ell_1^d\) implies, via Corollary~\ref{cor:ell1-wlsh}, a data-dependent LSH family for \(\metric\).

\begin{lemma}\label{lm:wkvg2lsh}
Let \(r > 0\), \(D\ge 1\), and \(\Delta>0\). Fix an approximation factor \(c \ge 64D\). Suppose that \(P\) is a \((cr,\frac12)\)-dispersed set of \(n\) points in a metric space \((\metric,\dist_\metric)\), and let \(\Delta\) be the diameter of \(P\). If \(f:\metric\to\ell_1^d\) (or \(f:\metric \to \ell_2^d\)) is an embedding with weak average distortion \(D\) with respect to \(P\), then there exists a \((r,p_1,p_2)\)-empirically sensitive \(\lsh(P)\) for \(P\) with 
\(
\frac{\log(1/p_1)}{\log(1/p_2)} \lesssim \frac{D}{c}
\)
and \(p_2 \ge 1-\frac{cr}{16D\Delta}\).

Moreover, assume that \(d\in \poly(\dim(\metric))\), and that \(f\) can be computed from \(P\) in time \(T\), and then stored in \(\poly(\dim(\metric))\) bits, and evaluated in \(\poly(\dim(\metric))\) time. Then a function \(h\sim \lsh(P)\) can be sampled in time \(O(T + \poly(n\dim(\metric)\Delta/r))\), evaluated in \(\poly(\dim(\metric)\Delta/r)\) time, and stored using using \(\poly(\dim(\metric)\Delta/r)\) bits. 
\end{lemma}
\begin{proof}
It suffices to prove the result for an embedding \(f\) into \(\ell_1^d\), since \(\ell_2^d\) embeds into \(\ell_1^{d'}\), for \(d'\lesssim d\), via an efficiently computable linear embedding with constant distortion~\cite{JS82}.
By rescaling, we can assume that \(\lip{f} = 1\), and, using Lemma~\ref{lm:disp-cdf}, we have
\[
\sup_{t \ge 0} t\Psi_{\ell_1^d}(f(P),t)
\ge \frac1D\sup_{t \ge 0} t\Psi_{\metric}(P,t)
\ge \frac{cr}D \Psi_{\metric}\left(P,cr\right)
\ge \frac{cr}{2D}.
\]
So, let \(t>0\) then be such that \(t\Psi_{\ell_1^d}(f(P),t) \ge \frac{cr}{2D}\). From Lemma~\ref{lm:disp-cdf}, we have that \(f(P)\) is \(\left(\frac{t}{2}, \frac{cr}{4Dt}\right)\)-dispersed. Moreover, the diameter of \(f(P)\) is at most \(\Delta\), so we can assume that \(t \le \Delta\). Without loss of generality, we can also assume that \(f(P) \subseteq [-\Delta,\Delta]^d\), as otherwise, we can shift \(f(P)\) so that one of its points is the origin.

Let us use \(\lsh_{\ell_1^d}\) to denote the data-dependent LHS family from Corollary~\ref{cor:ell1-wlsh} with parameters \(r\), \(c = \frac{t}{2r}\),
and \(\beta = \frac{cr}{4Dt}\). To sample \(h\sim \lsh(P)\), we compute the weak average distortion embedding \(f\), sample \(g\sim \lsh_{\ell_1^d}(f(P))\) and set \(h(x) = g(f(x))\). We have, by the guarantees of \(\lsh_{\ell_1^d}(f(P))\) and the fact that \(\lip{f}=1\), 
\begin{align*}
\dist_\metric(x,y) \le r \implies
\|f(x)-f(y)\|_1 \le r &\implies
\Pr_{g\sim\lsh_{\ell_1^d}(f(P))}[g(f(x)) = g(f(y))] \le 
1-\frac{16r}{t}.
\end{align*}
The last inequality is equivalent to \(\Pr_{h\sim\lsh(P)}[h(x) = h(y)] \le 1-\frac{16r}{t}\). This gives us the bound on the \(p_1\) parameter. For the \(p_2\) parameter analysis, let \(g\) be an arbitrary function in the support of \(\lsh_{\ell_1^d}(f(P))\), and let \(\bin\) be an arbitrary element of the range of \(\lsh_{\ell_1^d}(f(P))\). Then, 
\[
|\{x\in P: g(f(x)) = \bin\}| 
= |\{y\in f(P): g(y) = \bin\}| 
\le \left(1-\frac{cr}{16Dt}\right)n.
\]
In summary, \(\lsh(P)\) is \(\left(r,p_1, p_2\right)\)-empirically sensitive for \(p_1 = 1-\frac{16r}{t}\), \(p_2 = 1-\frac{cr}{16Dt}\le 1-\frac{cr}{16D\Delta}\). In particular,
\[
\frac{\ln(1/p_1)}{\ln(1/p_2)}
\le \frac{16r/t}{cr/(16Dt)}\left(1 -\ln\left(1-\frac{16r}{t}\right)\right)
\lesssim \frac{D}{c}.
\]
where the last inequality follows from \(\frac{16r}{t}\le \frac{32D}{c} \le 
\frac12\), which follows from \(t\Psi_{\ell_1^d}(f(P),t)\ge \frac{cr}{2D}\) and our assumption on \(c\).


The guarantee after ``moreover'' follows directly from Corollary~\ref{cor:ell1-wlsh} and the assumptions on \(f\).
\end{proof}

\section{From Average to Weak Average Distortion Embeddings}

As a final step towards proving Theorem~\ref{thm:avg2NNS}, in this section we show that we can construct weak average distortion embeddings using average distortion embeddings. As mentioned in Section~\ref{subsec:restech}, our argument is modelled after a similar connection between \(q\)- and \(r\)-average distortion embeddings, shown in~\cite{N14a,naor2019average}. We start with the theorem connecting the two notions of embeddings, which we prove in the remainder of the section.
\begin{theorem}\label{thm:avg2week}
Suppose that the metric space \((\metric,\dist_\metric)\) embeds with into a Banach space \((X,\|\cdot\|)\) with average distortion \(D\). Then, for any \(n\)-point set \(P\subseteq \metric\), there exists an embedding with weak average distortion \(D' \lesssim D(1 +\log D)\).  

Moreover, assume that for any point set \(Q\) of \(m\le n\) points in \(\metric\),  an embedding into \(X\) with average distortion \(D\) with respect to \(Q\) can be 
can be computed from \(Q\) in time \(T\), and then stored in \(\poly(\dim(\metric))\) bits, and evaluated in \(\poly(\dim(\metric))\) time. Then, 
for any \(n\)-point set \(P\subseteq \metric\), the  embedding into \(X\) with weak average distortion \(D'\) can be computed in time \(\poly(T + n\dim(\metric))\), stored in \(\poly(\dim(\metric))\) bits, and evaluated in time \(\poly(\dim(\metric))\).
\end{theorem}

Theorem~\ref{thm:avg2NNS} now follows  immediately from Lemmas~\ref{lm:main-ds}~and~\ref{lm:wkvg2lsh}, and Theorem~\ref{thm:avg2week}.

\medskip
We first handle an ``easy'' case, in which the distance function itself provides a good embedding into the line. The argument goes back to work by Rabinovich~\cite{Rabinovich08}.
\begin{lemma}\label{lm:avg2wkavg-easy}
Let \((\metric,\dist_\metric)\) be a metric space and let \(P\subseteq \metric\) be a subset containing $n$ points.
Define \(s(x) \coloneqq \min\{s: |P\cap B_\metric(x,s)| > \frac{n}{2}\}\) and let \(x^* \coloneqq  \arg \min_{x\in P} s(x)\). 
If 
\[
\sup_{t \ge s(x^*)}(t-s(x^*)) \frac{|P \setminus B_\metric(x^*,t)|}{n} \ge 
\alpha \sup_{t \ge 0} t \Psi_\metric(P,t),
\]
then the function \(f:\metric \to \R\) defined by \(f(x) = \dist_\metric(x^*,x)\) has weak average distortion \(\frac{1}{\alpha}\).
\end{lemma}
\begin{proof}
Suppose that \(t^*\) achieves
\(
\sup_{t \ge s(x^*)}(t-s(x^*)) |P \setminus B_\metric(x^*,t)|
\)
and let \(\beta^*  \coloneqq \frac{|P \setminus B_\metric(x^*,t^*)|}{n}\). 
We have
\begin{multline*}
|\{(x,y): x \in P \cap B_\metric(x^*,s(x^*)), y \in P \setminus B_\metric(x^*,t^*) \}|\\
=
|\{(y,x): x \in P \cap B_\metric(x^*,s(x^*)), y \in P \setminus B_\metric(x^*,t^*) \}| 
> 
\frac{\beta^*n^2}{2}.
\end{multline*}
For any 
\(
x \in P \cap B_\metric(x^*,s(x^*)), 
\)
and any
\(
y \in P \setminus B_\metric(x^*,t^*)
\)
we have \(|f(x) - f(y)| >  t^*-s(x^*)\), and, therefore, 
\(
\Psi_\R(f(P),t^*-s(x^*)) > \beta^*.
\)
This implies 
\[
\sup_{t\ge 0} t \Psi_\R(f(P),t) \ge
(t^*-s(x^*)) \Psi_\R(f(P),t^*-s(x^*)) 
> (t-s(x^*)) \beta^* \ge \alpha \sup_{t\ge 0}t \Psi_\metric(P,t).
\]
Since \(\lip{f}\le 1\) by the triangle inequality, the weak average distortion of \(f\) is at most \(\frac{1}{\alpha}\), as claimed.
\end{proof}

We will need the following standard technical lemma, which relates the weak-\(L_1\) and the \(L_1\) norm of the distance function over the uniform distribution on \(P\times P\).

\begin{lemma}\label{lm:wk-strong-avg}
Let \((\metric,\dist_\metric)\) be a metric space, and let \(P\subseteq \metric\) be an \(n\)-point set of diameter \(\Delta\). Then,
\[
\sup_{t \ge 0} t\Psi_\metric(P,t) 
< \frac{1}{n^2}\sum_{x \in P}\sum_{y\in P}\dist_\metric(x,y)
\le 
2\ln\left(\frac{2\Delta}{\frac{1}{n^2}\sum_{x \in P}\sum_{y\in P}\dist_\metric(x,y)}\right) \sup_{t \ge 0} t\Psi_\metric(P,t) .
\]
\end{lemma}
\begin{proof}
The first inequality follows since there are at least \(\Psi_\metric(P,t) n^2\) terms in the double sum on its right hand side that are greater than \(t\). 

Let \(a = \frac{1}{n^2}\sum_{x \in P}\sum_{y\in P}\dist_\metric(x,y) \) and \(b = \sup_{t \ge 0} t\Psi_\metric(P,t)\).
Towards the second inequality, observe that
\begin{align*}
a &= \int_{0}^\Delta \Psi_\metric(P,t) dt\\
&= \int_{0}^{\frac{a}{2}} \Psi_\metric(P,t) dt + \int_{\frac{a}{2}}^\Delta \Psi_\metric(P,t) dt\\
&\le \frac{a}{2} + \int_{\frac{a}{2}}^\Delta \frac{b}{t} dt
= \frac{a}{2} + b\ln\left(\frac{2\Delta}{a}\right).
\end{align*}
Re-arranging the terms proves the inequality.
\end{proof}

The following lemma is where we use the existence of an average distortion embedding.
\begin{lemma}\label{lm:avg2wkavg-hard}
Let \((\metric,\dist_\metric)\) and \((\calN,\dist_\calN)\) be metric spaces, and let \(P\subseteq \metric\) be an \(n\)-point set. Suppose that there exists an embedding \(f:\metric\to\calN\) with average distortion \(D\) such that the diameter of \(f(P)\) in \(\calN\) is at most \(\frac{C}{n^2}\sum_{x\in P}\sum_{y\in P} \dist_\metric(x,y)\). Then \(f\) has weak average distortion at most \(D' \le 2 D \log\left(\frac{2CD}{\lip{f}}\right).\)
\end{lemma}
\begin{proof}
Let $\Delta$ be the diameter of $f(P)$ in $\calN$.
Let, further, 
\begin{align*}
    &\overline{R}_{\metric} = \frac{1}{n^2}\sum_{x\in P}\sum_{y\in P}\dist_\metric(x,y),
    &&\overline{R}_{\calN} = \frac{1}{n^2}\sum_{x\in P}\sum_{y\in P}\dist_\calN(f(x),f(y)).
\end{align*}
By Lemma~\ref{lm:wk-strong-avg} and the assumption on diameter \(\Delta\) of \(f(P)\), we have
\begin{align*}
\sup_{t \ge 0} t\Psi_{\mathcal{N}}(f(P),t)&\geq \frac{\overline{R}_\calN}{2\ln\left(\frac{2\Delta}{\overline{R}_\calN}\right)}  
\geq \frac{\overline{R}_\calN}{2\ln\left(\frac{2C\overline{R}_\metric}{\overline{R}_\calN}\right)} 
\end{align*}
Now by the definition of average distortion and by Lemma~\ref{lm:wk-strong-avg}, it is clear that
\begin{align*}
    \overline{R}_\calN &\geq \frac{\lip{f}}{D}\overline{R}_\metric
    \geq \frac{\lip{f}}{D}\sup_{t \ge 0} t\Psi_\metric(P,t).
\end{align*}
Substituting into the previous inequality, we obtain
\begin{align*}
    \sup_{t \ge 0} t\Psi_{\mathcal{N}}(f(P),t)\geq \frac{\overline{R}_\calN}{2\ln\left(\frac{2C\overline{R}_\metric}{\overline{R}_\calN}\right)}&\geq
    \frac{\lip{f}}{2D\ln\left(\frac{2CD}{\lip{f}}\right)}\sup_{t \ge 0} t\Psi_\metric(P,t)\\
\end{align*}
This completes the proof.
\end{proof}

\begin{proof}[Proof of Theorem~\ref{thm:avg2week}]
For the rest of the proof we use the notation \(x^*\) and \(s(x^*)\) from Lemma~\ref{lm:avg2wkavg-easy}.

Let \(\alpha > 0\) be a small enough absolute constant, which we will choose later. 
Because of Lemma~\ref{lm:avg2wkavg-easy}, we can assume that
\begin{equation}\label{eq:concentration}
    \sup_{t \ge s(x^*)}(t-s(x^*)) \frac{|P \setminus B_\metric(x^*,t)|}{n} < 
\alpha \sup_{t \ge 0} t \Psi_\metric(P,t),
\end{equation}
as, otherwise, the lemma gives the required embedding, since the real line \(\R\) embeds in any Banach space. 

We first show the following claim.
\begin{claim}\label{cl:opt-t}
If the inequality \eqref{eq:concentration} holds, then \(\sup_{t \ge 0} t \Psi_\metric(P,t)\) is achieved for some \(t \le \frac{2s(x^*)}{1-4\alpha}\). 
\end{claim}
\begin{subproof}
Let us take some \(t' > \frac{2s(x^*)}{1-4\alpha}\) (where we assume \(\alpha \le \frac14\)). We will show that the supremum is not achieved at \(t'\). Note that, if \(\dist_\metric(x,y) > t'\), then by the triangle inequality, we must have one of \(\dist_\metric(x^*,x) > \frac{t'}{2}\) or \(\dist_\metric(x^*,y) > \frac{t'}{2}\). Therefore,
\[
\Psi_\metric(P,t') \le \frac{2|P \setminus B_\metric(x^*,\frac{t'}{2})|}{n}.
\]
Together with \eqref{eq:concentration} and the assumption \(t' > \frac{2s(x^*)}{1-4\alpha}\), this implies 
\[
t' \Psi_\metric(P,t')
\le 
\frac{2\alpha t'}{\frac{t'}{2} - s(x^*)}  \sup_{t \ge 0} t \Psi_\metric(P,t)
<   \sup_{t \ge 0} t \Psi_\metric(P,t).
\]  
This proves the claim. 
\end{subproof} 


Next we claim that 
\begin{equation}\label{eq:psi-sxstar}
\Psi_\metric\left(P,\frac{s(x^*)}{2}\right)
\ge  \frac{1}{2}.
\end{equation}
Indeed, by the definition of \(s(x^*)\), for any \(x\in P\) and any \(t < \frac{s(x^*)}{2}\), \(|P \cap B_\metric(x,2t)| \le \frac{n}{2}\). Lemma~\ref{lm:dispersed-P} then implies that \(P\) is \((t,\frac12)\)-dispersed for any \(t < \frac{s(x^*)}{2}\), 
and so, \eqref{eq:psi-sxstar} follows from Lemma~\ref{lm:disp-cdf}.

In the other direction, we claim that 
\begin{equation} \label{eq:sup-bound}
    \sup_{t \ge 0} t \Psi_\metric(P,t)
    \le
    \frac{2s(x^*)}{1-4\alpha}.
\end{equation}
To see this, let \(t^* \in \left[0,\frac{2s(x^*)}{1-4\alpha}\right]\) achieve \(\sup_{t \ge 0} t \Psi_\metric(P,t)\). We have
\[
\sup_{t \ge 0} t\Psi_\metric(P,t)
= t^* \Psi_\metric(P,t^*)  \le \frac{2s(x^*)}{1-4\alpha},
\]
where the last inequality follows because \(\Psi_\metric(P,t^*)\le 1\).


Let us define \(Q = P \cap B(x^*, 2s(x^*))\). By \eqref{eq:concentration} and \eqref{eq:sup-bound},
\begin{equation}\label{eq:Q-bound}
\frac{|P\setminus Q|}{n} < \frac{\alpha \sup_{t \ge 0} t\Psi_\metric(P,t)}{s(x^*)}
\le \frac{2\alpha}{1-4\alpha}\le \frac18,
\end{equation}
with the final inequality holding for any small enough \(\alpha\).
We then have
\[
\Psi_\metric\left(Q,\frac{s(x^*)}{2}\right) 
\ge \frac{n^2}{|Q|^2} 
\left(\Psi_\metric\left(P,\frac{s(x^*)}{2}\right) - \frac{2|P\setminus Q|}{n}\right)
> \frac14,
\]
where the final inequality follows by \eqref{eq:psi-sxstar} and \eqref{eq:Q-bound}. 
Therefore, by \eqref{eq:sup-bound} we have that
\begin{equation}\label{eq:psi-Q-lb}
   \frac{s(x^*)}{2} \Psi_\metric\left(Q,\frac{s(x^*)}{2}\right) 
\ge 
\frac{s(x^*)}{8} \ge 
\frac{(1-4\alpha)}{16}\sup_{t \ge 0} \Psi_\metric(P,t)
\ge \frac{1}{32}\sup_{t \ge 0} \Psi_\metric(P,t),
\end{equation}
with the last inequality again holding for small enough \(\alpha\). 
By Lemma~\ref{lm:wk-strong-avg}, we then have 
\[
\frac{1}{|Q|^2} \sum_{x\in Q}\sum_{y\in Q} \dist_\metric(x,y) >  
\frac{s(x^*)}{2} \Psi_\metric\left(Q,\frac{s(x^*)}{2}\right) \ge
\frac{\sup_{t \ge 0} \Psi_\metric(P,t)}{32}.
\]
At the same time,
the diameter of \(Q\) is \(4s(x^*)\) by construction, which is at most  \(16\sup_{t \ge 0} \Psi_\metric(P,t)\) by \eqref{eq:psi-sxstar}. Let us take \(f:\metric \to X\) to be an embedding with average distortion at most \(D\) with respect to \(Q\). The diameter of \(f(Q)\) in \(X\) is then at most 
\(
16\lip{f}\sup_{t \ge 0} \Psi_\metric(P,t).
\)
We can now use Lemma~\ref{lm:avg2wkavg-hard} with \(C = 512\lip{f} \) and get that \(f\) has weak average distortion at most \(D' \lesssim D(1+\log D)\) with respect to \(Q\). Note that by \eqref{eq:Q-bound},
\[
\sup_{t \ge 0} t\Psi_X(f(Q),t)
\le \frac{n^2}{|Q|^2} \sup_{t \ge 0} t\Psi_X(f(P),t)
< \left(\frac{8}{7}\right)^2 \sup_{t \ge 0} t\Psi_X(f(P),t).
\]
We then have 
\[
\sup_{t\ge 0} t\Psi_X(f(P),t) \gtrsim
\sup_{t \ge 0} t\Psi_X(f(Q),t)
\gtrsim \frac{\lip{f}}{D(1+\log D)} \sup_{t \ge 0} t\Psi_\metric(Q,t)
\gtrsim 
\frac{\lip{f}}{D(1+\log D)} \sup_{t \ge 0} t\Psi_\metric(P,t),
\]
with the final inequality implied by \eqref{eq:psi-Q-lb}. This shows that \(f\) has weak average distortion \(D' \lesssim D(1+\log D)\), as required. 

For the statement after ``moreover'' we just observe that the embedding is either given by \(\dist_\metric(x^*,x)\), which we assume can be evaluated in time \(\poly(\dim(\metric))\), or is given by an average distortion \(D\) embedding with respect to \(Q\), which we also assume can be computed and evaluated in the required time.
\end{proof}
\section{Efficient Average Distortion Embeddings }

In this section we present our constructions of explicit average distortion embeddings. We first give a general result on average distortion embeddings derived from bi-H\"older homeomorphisms between spheres of Banach spaces. We then apply this general result to \(\ell_p\) and Schatten-\(p\) spaces in the subsequent subsections.

\subsection{Average Distortion Embeddings from Bi-H\"older Homeomorphisms}

We first give a general construction of embeddings with bounded average distortion using homeomorphisms between spheres of normed spaces. This construction uses techniques that were used previously to prove inequalities between non-linear Rayleigh quotients in~\cite{daher}, and go back to Matou\v{s}ek's extrapolation theorem~\cite{M97}. Here we show that these techniques can be used to directly prove the existence of average distortion embeddings via an explicit construction.

We first show that homeomorphisms between spheres can be radially extended to the entire normed spaces while retaining their continuity properties. The lemma below was also shown in~\cite{daher} but with worse constants. 
\begin{lemma}\label{lm:biholder-extension}
Let \((X, \|\cdot\|_X)\) and \((Y,\|\cdot\|_Y)\) be Banach spaces with unit spheres, respectively, $S_X$ and $S_Y$. Let \(\alpha, \beta \in (0,1]\). Let \(f:S_X \to S_Y\) be a function that, for any \(x,y \in S_X\) satisfies
\[
\frac1L \|x-y\|_X^{1/\beta} \le \|f(x) - f(y)\|_Y \le K \|x-y\|_X^\alpha.
\]
Then the function \(\tilde{f}:X \to Y\) defined by
\(
\tilde{f}(x) = \|x\|_X^\alpha f\left(\frac{x}{\|x\|_X}\right)
\) for $x\neq 0$,
and $\tilde{f}(0)=0$
satisfies the following for any \(x,y\in X\):
\begin{align}
    \|\tilde{f}(x) - \tilde{f}(y)\|_Y &\le
(1 + 2^\alpha K)\|x-y\|_X^\alpha\label{eq:extension-ub}\\   
\|\tilde{f}^{-1}(x) - \tilde{f}^{-1}(y)\|_X
&\le \left(\frac{1}{\alpha\beta} + 2^\beta L^\beta\right) \|x-y\|_Y^\beta \max\{\|x\|_Y, \|y\|_Y\}^{\frac1\alpha-\beta }\label{eq:extension-lb}
\end{align}
Moreover, \(\|\tilde{f}(x)\|_Y = \|x\|_X^\alpha\) for all \(x\in X\).
\end{lemma}
\begin{proof}
The equality after ``moreover'' is obvious from the fact tht \(f\) is a map between the unit spheres.
Both \eqref{eq:extension-ub} and \eqref{eq:extension-lb} follow from the following more general fact. Let \(U,V\) be Banach spaces, and let \(g:S_U \to S_V\) and \(\omega \in (0,1]\) be such that for any \(x,y\in S_U\) we have
\[
\|g(x) - g(y)\|_V \le C \|x-y\|_U^\omega.
\]
Suppose that \(p \ge \omega\). Then the degree \(p\) radial extension \(g:U \to Y\) defined by \(\tilde{g}(x) = \|x\|_U^p g\left(\frac{x}{\|x\|_U}\right)\) and \(\tilde{g}(0)=0\) satisfies 
\begin{equation}\label{eq:radial-ext-gen}
\|\tilde{g}(x) - \tilde{g}(y)\|_V
\le \left(\frac{p}{\omega} + 2^\omega C\right)\cdot\|x-y\|_U^\omega\cdot \max\{\|x\|_U, \|y\|_U\}^{p-\omega }.
\end{equation}
Then \eqref{eq:extension-ub} follows by setting \(U = X\), \(V=Y\), \(g = f\), \(\omega = \alpha\), \(p = \alpha\), and \(C = K\). The other estimate  \eqref{eq:extension-lb} follows by setting \(U=Y\), \(V=X\),\(g = f^{-1}\), \(\omega = \beta\), \(p = \frac1\alpha\), and \(C = L^\beta\). To see that \eqref{eq:radial-ext-gen} applies to this second case, notice that \(\tilde{f}^{-1}(y) = \|y\|_Y^{1/\alpha} f^{-1}\left(\frac{y}{\|y\|_Y}\right)\) for \(y\neq 0\).

Clearly \eqref{eq:radial-ext-gen} holds if either \(x\) or \(y\) is \(0\), so we assume that both are nonzero. Note first that both the left and the right hand side are homogeneous in \(x\) and \(y\) of degree \(p\), so we can, without loss of generality, assume that \(0< \|x\|_U \le \|y\|_U = 1\). Under this assumption, we just need to show that 
\begin{equation*}
\|\tilde{g}(x) - \tilde{g}(y)\|_V 
\le \left(\frac{p}{\omega} + 2^\omega C\right)\|x-y\|_U^\omega.
\end{equation*}
Let us use the notation \(\hat{x} = \frac{x}{\|x\|_U}\), so that \(\tilde{g}(x) = \|x\|_U^p g(\hat{x})\). Moreover, \(\tilde{g}(y) = g(y)\). We then have
\begin{align}
    \|\tilde{g}(x) - \tilde{g}(y)\|_V 
    &\le (1 - \|x\|_U^p)\|g(\hat{x})\|_U + \|g(\hat{x}) - g(y)\|_V\notag\\
    &\le 1 - \|x\|_U^p + C\|\hat{x} - y\|_U^\omega.\label{eq:ext-1}
\end{align}
Using Bernoulli's inequality \(1-t^r \le r(1-t)\), which holds for any \(r \ge 1\) and any \(t \le 1\), and the inequality \(1 - t^\gamma \le (1-t)^\gamma\), which holds for any \(t, \gamma \in [0,1]\), we get 
\begin{equation}\label{eq:ext-2}
1 - \|x\|_U^p \le \frac{p}{\omega}(1-\|x\|_U^{\omega})
\le
\frac{p}{\omega}(1-\|x\|_U)^{\omega}.
\end{equation}
By the triangle inequality, \(\|x-y\|_U \ge \|y\|_U - \|x\|_U = 1-\|x\|_U\), and, moreover,
\begin{equation}\label{eq:ext-3}
\|\hat{x} -y\|_U \le \|x-y\|_U + \|\hat{x} - x\|_U 
= 
\|x-y\|_U + 1-\|x\|_U 
\le 2\|x-y\|_U.
\end{equation}
Combining inequalities \eqref{eq:ext-1}--\eqref{eq:ext-3}, we get
\[
\|\tilde{g}(x) - \tilde{g}(y)\|_V 
\le \frac{p}{\omega}(1-\|x\|_U)^{\omega} + 2^\omega C\|x - y\|_U^\omega
\le \left(\frac{p}{\omega} + 2^\omega C\right) \|x-y\|_U,
\]
as we needed to prove.
\end{proof}

Recall that a median of a set of \(n\) points \(P\) in a metric space \(\metric\) is any point \(y \in \metric\) that minimizes \(\frac1n \sum_{x\in P}\dist_\metric(x,y)\). We generalize this definition by allowing approximation, and also by allowing the distance function to be raised to a power \(p\).

\begin{definition}
We say that a point \(y\) in a metric space \(\metric\) is a \((C,\varepsilon)\)-approximate median of a finite point set \(P \subseteq \metric\) if 
\[
\frac1n\sum_{x\in P}\dist_\metric(x,y) 
\le C\min_{z\in\metric}\frac{1}{n} \sum_{x\in P}\dist_\metric(x,z)
+ \varepsilon.
\]
More generally, we say that \(y\) is a \((C,\varepsilon)\)-approximate \(q\)-mean if 
\[
\frac1n\sum_{x\in P}\dist_\metric(x,y)^q 
\le C^q\min_{z\in\metric}\frac{1}{n} \sum_{x\in P}\dist_\metric(x,z)^q
+ \varepsilon^q.
\]
A \((1,0)\)-approximate \(q\)-mean of \(P\) is just called a \(q\)-mean of \(P\).
\end{definition}

We will need an easy technical lemma which gives bounds on the \(q\)-mean objective and shows that means (in the standard sense of averaging points in a vector space) are good approximate \(q\)-means.

\begin{lemma}\label{lm:mean-to-median}
Suppose that \(q\ge 1\) and that \(P\) is an \(n\)-point set in a metric space \((\metric, \dist_\metric)\). We have the inequalities
\begin{equation*}
\frac{1}{2^q n^2} \sum_{x\in P}\sum_{y\in P}\dist_\metric(x,y)^q
\le \min_{z\in \metric}\frac1n\sum_{x\in P}\dist_\metric(x,z)^q
\le \frac{1}{n^2} \sum_{x\in P}\sum_{y\in P}\dist_\metric(x,y)^q.
\end{equation*}

Moreover, suppose that \(\metric\) is defined by a Banach space \((X,\|\cdot\|)\). If \(z\in X\) satisfies
\[
\left\|z - \frac1n \sum_{x\in P}x\right\|\le  \varepsilon
\]
then \(z\) is a \((2^{2 - \frac1q},2^{1-\frac1q}\varepsilon)\)-approximate \(q\)-mean of \(P\).
\end{lemma}
\begin{proof}
Let \(u\) be a \(q\)-mean of \(P\). By the triangle inequality, and H\"older's inequality, we have the lower bound
\begin{equation}\label{eq:med-lb}
\frac{1}{2^q n^2} \sum_{x\in P}\sum_{y\in P}\dist_\metric(x,y)^q
\le \frac{1}{2^q n^2} \sum_{x\in P}\sum_{y\in P}(\dist_\metric(x,u) + \dist_\metric(y,u))^q
\le \frac1n\sum_{x\in P}\dist_\metric(x,u)^q.
\end{equation}
To prove the second inequality, observe that 
\[
\min_{y\in \metric}\frac1n\sum_{x\in P}\dist_\metric(x,y)^q
\le 
\min_{y\in P}\frac1n\sum_{x\in P}\dist_\metric(x,y)^q
\le \frac{1}{n^2} \sum_{y\in P}\sum_{x\in P}\dist_\metric(x,y)^q,
\]
and changing the order of summation (or simply renaming the variables) finishes the proof.

To prove the claim after ``moreover'', observe that, by the assumption on \(z\), the triangle inequality, H\"older's inequality, and Jensen's inequality applied to the convex function $\|\cdot\|^q$, we have
\begin{align*}
\frac1n\sum_{x\in P}\|x-z\|^q
&\le 
\frac{2^{q-1}}{n}\sum_{x\in P}\left\|x - \frac1n \sum_{y\in P}y\right\|^q + 2^{q-1} \varepsilon^q
\le \frac{2^{q-1}}{n^2} \sum_{x\in P}\sum_{y\in P}\left\|x - y\right\|^q+ 2^{q-1} \varepsilon^p.
\end{align*}
Combining this inequality with \eqref{eq:med-lb} finishes the proof. 
\end{proof}

The next lemma is our main tool for constructing explicit average distortion embeddings. Recall that, for \(\alpha \in (0,1]\) the \(\alpha\)-snowflake of a metric space \((\metric, \dist_\metric)\) is the metric space \(\metric^\alpha\) on the same ground set, with distance function \(\dist_\metric(x,y)^\alpha\).

\begin{lemma}\label{lm:embedding-gen}
Let \((X, \|\cdot\|_X)\) and \((Y,\|\cdot\|_Y)\) be Banach spaces, and let \(\alpha \in (0,1]\), and \(p\ge 1\). Let \(f:S_X \to S_Y\) be a function that, for any \(x,y \in S_X\) satisfies
\[
\|f(x) - f(y)\|_Y \le K \|x-y\|_X^\alpha,
\]
and let \(\tilde{f}:X\to Y\) be defined as in Lemma~\ref{lm:biholder-extension}.
Let \(P\subseteq X\) be an \(n\) point set, let \(t \in X\), and define \(g:X^\alpha \to Y\) by \(g(x) = \tilde{f}(x-t)\). Suppose that one of the following conditions is satisfied for some \(C \ge 1\) and \(\varepsilon \le \frac12\left(\frac{1}{2n^2} \sum_{x\in P}\sum_{y\in P}\|x-y\|_X^{q\alpha}\right)^\frac1q\)
\begin{enumerate}
    \item \(0\) is a \((C,\varepsilon)\)-approximate \(q\)-mean for \(g(P)\);
    \item \(\left\|\frac1n \sum_{x\in P}g(x)\right\|_Y\le 2^{\frac1q - 1}\varepsilon\).
\end{enumerate}
Then \(g\) is an embedding of the \(\alpha\)-snowflake \(X^\alpha\) into \(Y\) with average distortion at most \(D\) with respect to \(P\), where \(D \lesssim C(1+K)\) if the first condition is satisfied, and \(D\lesssim 1+K\) if the second condition is satisfied.
\end{lemma}
The proof is similar in spirit to bounds on non-linear Rayleigh quotients proved in~\cite{daher}.
\begin{proof}
By Lemma~\ref{lm:biholder-extension}, we have 
\begin{equation}\label{eq:embed-lip}
    \lip{g} = \lip{\tilde{f}} \le 1 + 2^{\alpha}K.
\end{equation}
It remains to calculate a lower bound on the average distance between points in $g(P)$. By the first inequality in Lemma~\ref{lm:mean-to-median}, and the fact that, by Lemma~\ref{lm:biholder-extension}, \(\|g(x)\|_Y = \|\tilde{f}(x-t)\|_Y = \|x-t\|_X^\alpha\) for any \(x\in X\), we have
\begin{align}
    \frac{1}{n^2}\sum_{x\in P}\sum_{y\in P}\|x-y\|_X^{q\alpha} 
    &\le \frac{2^q}{n}\sum_{x\in P}\|x-t\|_X^{q\alpha}
    =  \frac{2^q}{n}\sum_{x\in P}\|g(x)\|_Y^q \label{eq:sphere}
\end{align}
 We now consider the first case, i.e., that \(0\) is a \((C,\varepsilon)\)-approximate $q$-mean. The second case reduces to the first case since, by Lemma~\ref{lm:mean-to-median}, if \(\left\|\frac1n \sum_{x\in P}g(x)\right\|_Y\le \varepsilon\), then \(0\) is a \((2^{2-\frac1q},\varepsilon)\)-approximate median of \(g(P)\).

It follows from Lemma~\ref{lm:mean-to-median} that 
\begin{align}
   \frac{1}{n}\sum_{x\in P}\left\|g(x)\right\|_Y^q &\leq C^q\min_{z\in Y}\frac{1}{n} \sum_{x\in P}\|g(x) - z\|_Y + \varepsilon^q\\
   &\leq \frac{C^q}{n^2}\sum_{x\in P}\sum_{y\in P}\|g(x)-g(y)\|^q_Y + \frac{1}{2^{q+1}n^2}\sum_{x\in P}\sum_{y\in P}\|x-y\|_X^{q\alpha}. \label{eq:avg-dist-med}
\end{align}
Combining (\ref{eq:sphere}) and (\ref{eq:avg-dist-med}), we obtain
\[
\sum_{x\in P}\sum_{y\in P}\|x-y\|_X^{q\alpha} \leq 2^{q+1}C^q \sum_{x\in P}\sum_{y\in P}\|g(x)-g(y)\|^q_Y,
\]
from which, together with (\ref{eq:embed-lip}), we conclude that \(g\) is an embedding of the \(\alpha\)-snowflake \(X^\alpha\) into \(Y\) with \(q\)-average distortion at most \(D\) with respect to \(P\), where \(D \leq 2^{1+\frac1q} C(1+2^\alpha K) \lesssim C(1+K)\).
\end{proof}

In order to use Lemma~\ref{lm:embedding-gen}, we need to find some \(t\in X\) that satisfies one of the two assumptions in the lemma. A general method for establishing the existence of such a \(t\) was proposed in \cite{spectral,daher,naor2019average}, and relies on the following lemma. For a proof, see Lemma~45 from \cite{naor2019average}.

\begin{lemma}\label{lm:surjective}
For any finite dimensional Banach space \((X,\|\cdot\|)\), and any continuous function $h:X \to X$ such that
\[
\lim_{M\to \infty} \inf_{t: \|t\| \ge M}(\|t\| - \| h(t)-t\|) = \infty,
\]
we have that $h$ is surjective.
\end{lemma}

Using Lemma~\ref{lm:surjective}, we can show that there exists a \(t\) so that \(0\) is the mean of \(g(P)\). The argument is essentially identical to arguments in~\cite{spectral,daher}, but, since the result was not stated in the general form given below, we include a proof.
\begin{lemma}\label{lm:center}
Under the assumptions and notation of Lemma~\ref{lm:biholder-extension}, for any \(n\)-point set \(P \subseteq X\) there exists some \(t\in X\) such that \(\frac1n \sum_{x\in P}\tilde{f}(x-t) =0\). Moreover, any such \(t\) must satisfy 
\(
\|t\|_X \le 
\left(\frac{M}{n} \sum_{x\in P}\|x\|_X^\alpha\right)^{\frac1\alpha},
\)
 for a constant \(M\) that only depends on \(K,L,\alpha,\beta\).
\end{lemma}
\begin{proof}
We will show that the function $h:X \to X$ defined by $h(t) = \tilde{f}^{-1}\left(\frac1n \sum_{x\in P} \tilde{f}(x-t)\right)$ satisfies
\begin{equation}\label{eq:h-bound}
\|h(t)-t\|_X \le CR^\beta(\|t\|_X^\alpha+R)^{\frac1\alpha - \beta},
\end{equation}
for \(C\) that only depends on \(K,L,\alpha,\beta\), and \(R=\frac1n \sum_{x\in P}\|x\|^\alpha_X\). This means that, for large \(\|t\|_X\), \(\|h(t)-t\|_X\) is dominated by \(CR^\beta \|t\|_X^{1-\alpha \beta}\), and, therefore, 
\(
\|t\|_X - \|h(t)-t\|_X \to \infty
\) 
as \(\|t\|_X \to \infty\).  Lemma~\ref{lm:surjective} then implies that for some \(t\in X\) we have \(h(t) = 0\), which is equivalent to \(\frac1n \sum_{x\in P} \tilde{f}(x-t)=0\) since \(\tilde{f}\) is a bijection and \(\tilde{f}(0)=0\).

The claim after ``moreover'' also follows from \eqref{eq:h-bound}. Indeed, suppose that \(\|t\|_X^\alpha > M R\). Then, for any large enough \(M\) that only depends on \(C, \alpha,\beta\), \(\|h(t)-t\|_X < C \frac{1}{M^\beta} (1+\frac{1}{M})^{\frac1\alpha - \beta}\|t\|_X < \|t\|_X\). By the triangle inequality we then have \(\|h(t)\|_X > 0\), so no such \(t\) can satisfy \(\frac1n \sum_{x\in P} \tilde{f}(x-t)=0\).

Next, we prove \eqref{eq:h-bound}. Observe that, by  \eqref{eq:extension-lb},
\begin{align*}
  \|h(t)-t\|_X &= \left\|\tilde{f}^{-1}\left(\frac1n \sum_{x\in P} \tilde{f}(x-t)\right) - \tilde{f}^{-1}(\tilde{f}(t))\right\|_X\\
 &\leq \left(\frac{1}{\alpha\beta} + (2L)^\beta\right)\cdot \left\|\frac1n \sum_{x\in P} \left( \tilde{f}(x-t) - \tilde{f}(t)\right)\right\|_Y^{\beta} \cdot\  \max\left\{\left\|\frac1n\sum_{x\in P} \tilde{f}(x-t)\right\|_Y, \|\tilde{f}(t)\|_Y \right\}^{\frac1\alpha-\beta}.\\
\end{align*}
Next, note that $\|\tilde{f}(t)\|_Y = \|t\|_X^\alpha$ and, further, 
\[
\left\|\frac1n\sum_{x\in P} \tilde{f}(x-t)\right\|_Y \leq \frac{1}{n}\sum_{x\in P} \|\tilde{f}(x-t)\|_Y = \frac{1}{n}\sum_{x\in P} \|x-t\|_X^\alpha\leq \frac{1}{n}\sum_{x\in P} (\|x\|_X^\alpha+\|t\|_X^\alpha) = \|t\|_X^\alpha + R.
\]
Thus,
\[
\max\left\{\left\|\frac1n\sum_{x\in P} \tilde{f}(x-t)\right\|_Y,\|\tilde{f}(t)\|_Y \right\}^{\frac1\alpha-\beta} \leq (\|t\|_X^\alpha + R)^{\frac1\alpha-\beta}.
\]
Continuing our calculation, we obtain, using \eqref{eq:extension-ub},
\begin{align*}
    \|h(t)-t\|_X &\leq \left(\frac{1}{\alpha\beta} + (2L)^\beta\right)\cdot\left(\frac1n \sum_{x\in P} \left\| \tilde{f}(x-t) - \tilde{f}(t)\right\|_Y\right)^{\beta}\cdot\  (\|t\|_X^\alpha + R)^{\frac1\alpha-\beta}\\
 &=  \left(\frac{1}{\alpha\beta} + (2L)^\beta\right)\cdot\left(\frac1n \sum_{x\in P} (1+2^\alpha K) \|x\|_X^\alpha\right)^\beta \cdot\  (\|t\|_X^\alpha + R)^{\frac1\alpha-\beta}\\
  &= \left(\frac{1}{\alpha\beta} + (2L)^\beta\right)\cdot((1+2^\alpha K)R)^\beta \cdot\  (\|t\|_X^\alpha + R)^{\frac1\alpha-\beta}.
\end{align*}
Hence, inequality \eqref{eq:h-bound} follows for $C = \left(\frac{1}{\alpha\beta} + (2L)^\beta\right)(1+2^\alpha K)^\beta$.
\end{proof}

While Lemma~\ref{lm:center} is very general, it does not readily give rise to an efficient algorithm to find \(t\). The proof of Lemma~\ref{lm:surjective} in \cite{spectral,naor2019average} is existential, and relies on a homological degree argument of the type used to prove Brouwer's fixed point theorem. Identifying general cases in which we can give an algorithmic proof of Lemma~\ref{lm:center} is an interesting open problem. In the following sections, we give alternative algorithmic methods for finding a good center \(t\) when \(X\) is \(\ell_p^d\) or the Schatten-\(p\) norm for \(1\le p \le 2\). 

For ease of reference, below we state the general existence result implied by Lemmas~\ref{lm:embedding-gen}~and~\ref{lm:center}, giving average distortion embeddings from radially extended and shifted bi-H\"older homeomorphisms. 
\begin{theorem}\label{thm:embegging-gen}
Let \((X, \|\cdot\|_X)\) and \((Y,\|\cdot\|_Y)\) be Banach spaces, and let \(\alpha, \beta \in (0,1]\). Let \(f:S_X \to S_Y\) be a function that, for any \(x,y \in S_X\) satisfies
\[
\frac1L \|x-y\|_X^{1/\beta} \le \|f(x) - f(y)\|_Y \le K \|x-y\|_X^\alpha.
\]
Let the function \(\tilde{f}:X \to Y\) be defined by
\(
\tilde{f}(x) = \|x\|_X^\alpha f\left(\frac{x}{\|x\|_X}\right)
\)
for \(x \neq 0\), and \(\tilde{f}(0)=0\).
Then, for any \(n\)-point set \(P \subseteq X\) there exists some \(t\in B_X(0,R)\), \(R \le \left(\frac{M}{n} \sum_{x\in P}\|x\|_X^\alpha\right)^{\frac1\alpha}
\)
 for a constant \(M\) that only depends on \(K,L,\alpha,\beta\),  such that the function \(g:X^\alpha \to Y\) defined by \(g(x) = \tilde{f}(x-t)\) is an embedding of the \(\alpha\)-snowflake \(X^\alpha\) into \(Y\) with average distortion at most \(D\lesssim 1+K\). 
\end{theorem}

\subsection{Embeddings and Data Structures for \(\ell_p\) Spaces}

Let us recall the classical Mazur map $M_{p,q}:\ell_p^d \to \ell_q^d$ given by 
\[
\forall x \in \R^d \ \ \forall i \in [d]: 
M_{p,q}(x)_i = \sign(x_i) |x_i|^{\frac{p}{q}}.
\]
These maps are homeomorphisms between spheres: for an $x \in \R^d$, $\|M_{p,q}(x)\|_q^q = \|x\|_p^p$, and, for $p\ge q$ and any $x,y\in \R^d$ we have
\begin{equation*}
    \frac{ \|x-y\|_p^{p/q}}{2^{(p-q)/q}}
    \le 
    \|M_{p,q}(x) - M_{p,q}(y)\|_q \le \frac{p}{q} \|x-y\|_p \cdot (\|x\|_p^p + \|y\|_p^p)^{\frac1q-\frac1p}.
\end{equation*}
See Section~5.1.~of~\cite{N14a} for derivations of these classical inequalities. In particular, for any \(x,y\in S_{\ell_p^d}\) and any $p\ge q\ge 1$, we have
\begin{equation}\label{eq:mazur-cont}
   \frac{ \|x-y\|_p^{p}}{2^{(p-q)/q}}
    \le 
    \|M_{p,q}(x) - M_{p,q}(y)\|_q 
    \le 2^{\frac1q - \frac1p} \frac{p}{q} \|x-y\|_p.
\end{equation}

For any \(p \ge q\ge 1\), we will show an embedding \(\ell_p^d\) into \(\ell_q^d\) with \(q\)-average distortion \(O(\frac{p}{q})\). The \(q=1\) case is what we need for our data structures, while the \(q=2\) case answers a question of Naor~\cite{N14a}.
We will use the re-scaled function \(\tilde{M}_{p,q}\), as defined in Lemma~\ref{lm:biholder-extension}. I.e., for any \(x\in \R^d\) we define
\[
\tilde{M}_{p,q}(x)=
\|x\|_p \cdot M_{p,q}\left(\frac{x}{\|x\|_p}\right)
= \|x\|_p^{1 - \frac{p}{q}} M_{p,q}(x).
\]
By Lemma~\ref{lm:biholder-extension}, \(\tilde{M}_{p,q}\) is a Lipschitz function with Lipschitz constant \(O(\frac{p}{q})\), and we will see that an appropriate shift of it satisfies the conditions of Lemma~\ref{lm:embedding-gen}, and gives us the desired embedding. The shift we will use is very simple, and is given in the following lemma.

\begin{lemma}\label{lm:med-shift}
Suppose that \(P \subseteq \R^d\) is an \(n\)-point set and that \(t\in \R^d\) is such that \(t_i\) is a median of the (multi-)set \(\{x_i: x\in P\}\). I.e., suppose that
\begin{equation}\label{eq:med-ell1}
    \forall i\in [d]:\ \ 
    |\{x\in P: x_i < t_i\}| = |\{x\in P: x_i > t_i\}|.
\end{equation}
Then \(0\) is a median of \(\{\tilde{M}_{p,1}(x-t): x\in P\}\) as a subset of \(\ell_1^d\). Moreover, \(0\) is a \((2^{1+\frac1q},0)\)-approximate \(q\)-mean of \(\{\tilde{M}_{p,q}(x-t): x\in P\}\) as a subset of \(\ell_q^d\).
\end{lemma}
\begin{proof}
The key observation is that, for any \(p\) and \(q\), any \(x\in \R^d\), and any \(i\in[d]\),
\[
\mathrm{sign}(\tilde{M}_{p,q}(x)_i) 
=\mathrm{sign}({M}_{p,q}(x)_i)
=\mathrm{sign}(x_i).
\]
Therefore, 
\[
|\{x\in P: \tilde{M}_{p,q}(x-t)_i < 0\}| 
= |\{x\in P: x_i < t_i\}| = |\{x\in P: x_i > t_i\}| 
= |\{x\in P: \tilde{M}_{p,q}(x-t)_i > 0\}|.
\]
This means that \(0\) is a median of \(\{\tilde{M}(x-t)_i: x\in P\}\) for any \(i\in [d]\), which is equivalent to \(0\) being a median of \(\{\tilde{M}_{p,1}(x-t): x\in P\}\) with respect to the \(\ell_1^d\) norm, i.e., to \(0 \in \arg\min_y \frac1n \sum_{x\in P}\|\tilde{M}_{p,1}(x-t)-y\|_1\).

It remains to prove the lemma for \(q > 1\). We observe that, for any \(i \in [d]\), 
\[
\frac{1}{n} \sum_{x\in P} |\tilde{M}_{p,q}(x)_i|^q \le \frac{2}{n^2}\sum_{x\in P} \sum_{y\in P} |\tilde{M}_{p,q}(x)_i-\tilde{M}_{p,q}(y)_i|^q,
\]
since, for each \(x \in P\), 
\begin{multline*}
    |\{y \in P: |\tilde{M}_{p,q}(x)_i-\tilde{M}_{p,q}(y)_i| \ge  |\tilde{M}_{p,q}(x)_i|\}| \\
\ge |\{y \in P: y_i = 0 \text{ or } \mathrm{sign}(\tilde{M}_{p,q}(x)_i) \neq \mathrm{sign}(\tilde{M}_{p,q}(y)_i)\}| \ge \frac{n}{2}.
\end{multline*}
Then, by Lemma~\ref{lm:mean-to-median}, 
\[
\frac{1}{n} \sum_{x\in P} \|\tilde{M}_{p,q}(x)\|_q^q 
\le 
\frac{2}{n^2}\sum_{x\in P} \sum_{y\in P} \|\tilde{M}_{p,q}(x)-\tilde{M}_{p,q}(y)\|_q^q
\le 
2^{q+1} \min_{z\in \R^d} \frac{1}{n} \sum_{x\in P} \|\tilde{M}_{p,q}(x)-z\|_q^q,
\]
showing that \(0\) is a \((2^{1+\frac1q},0)\)-approximate \(q\)-mean, as claimed.
\end{proof}

We can now re-state our main embedding theorem for \(\ell_p^d\).

\embellp*
\begin{proof}
Follows directly from \eqref{eq:mazur-cont},  Lemma~\ref{lm:med-shift}, and Lemma~\ref{lm:embedding-gen}.
\end{proof}

Theorem~\ref{thm:embedding-ellp} (in the case \(q = 1\)), together with Theorem~\ref{thm:avg2week} directly implies the existence of efficiently computable embedding of \(\ell_p^d\) into \(\ell_1^d\) 
with weak average distortion \(O(p\log p)\). We can remove the extra logarithmic factor with the help of the following lemma.

\begin{lemma}\label{lm:Median_not_to_far}
Let \(p \ge 1\), and let \(P\) be an $n$ point set in \(\R^d\). Let \(t\) satisfy \eqref{eq:med-ell1}. Then
\[
\max_{x\in P}\|x-t\|_p \lesssim \max_{x\in P}\left\|x - \frac1n \sum_{y\in P}y\right\|_p \le \max_{x,y\in P}\|x-y\|_p.
\]
\end{lemma}
\begin{proof}
The second inequality follows from the triangle inequality, so we only need to prove the first one. We prove the following claim first.
\begin{claim}
For any $a_1, a_2, ..., a_n\in \mathbb{R}$, and a median $m$ of these $n$ numbers, we have
$$|m|^p\leq \frac{2}{n}\sum_i{|a_i|^p}$$
\end{claim}
\begin{proof}
Since $m$ is a median,  at least $\lceil{\frac{n}{2}\rceil}$ of the \(a_i\) have absolute value greater than or equal to that of  $m$. This means that
$$\sum_i{|a_i|^p}\geq \left\lceil\frac{n}{2}\right\rceil|m|^p,$$
which implies the claim.
\end{proof}
Coming back to the proof of the lemma, by the triangle inequality we have
\begin{align*}
    \max_{x\in P}\|x-t\|_p\leq \max_{x\in P}\left\|x - \frac1n \sum_{y\in P}y\right\|_p
    +\left\|t - \frac1n \sum_{y\in P}y\right\|_p
\end{align*}
By applying the claim on each coordinate \(i\) to the set 
\(
\left\{x_i - \frac1n\sum_{y\in P}y_i: x\in P\right\}
\), and the median
\(
t_i -\frac1n \sum_{y\in P}y_i,
\)
we obtain
\begin{align*}
    \left\|t - \frac1n \sum_{y\in P}y\right\|^p_p&\leq \sum_{i=1}^d{\frac{2}{n}\sum_{x\in P}{\left|x_i-\frac{1}{n}\sum_{y\in P}{y_i}\right|^p}}\\
    &= \frac2n\sum_{x\in P}\left\|x-\frac{1}{n}\sum_{y\in P}{y}\right\|_p^p
    \le 2\max_{x\in P}{\left\|x-\frac{1}{n}\sum_{y\in P}{y}\right\|_p^p}.
\end{align*}
This implies that
$$\max_{x\in P}\|x-t\|_p\leq (1 + 2^{\frac1p})\max_{x\in P}\left\|x - \frac1n \sum_{y\in P}y\right\|_p$$
which completes the proof.
\end{proof}

\begin{theorem}\label{thm:wkembed-ellp}
For any $p \ge 1$, and any $n$-point set $P$ in $\R^d$, there exists an embedding \(f:\ell_p^d \to \ell_1^d\) with weak average distortion \(D \lesssim p\) with respect to \(P\), such that \(f\) can be computed in time \(\poly(nd)\), stored in \(poly(d)\) bits, and evaluated in time \(\poly(d)\).
\end{theorem}
\begin{proof}
The proof follows from Theorem~\ref{thm:embedding-ellp} and the proof of Theorem~\ref{thm:avg2week} with the following modification. In the final step of the proof, we take the embedding \(f=g\) with average distortion \(D \lesssim p\) with respect to \(Q\) given by Theorem~\ref{thm:embedding-ellp}. This embedding maps \(x\in Q\) to \(\tilde{M}_{p,1}(x-t)\) where \(t\) is as defined by \eqref{eq:med-ell1}. By construction, we know that \(Q\) has diameter at most \(16 \sup_{t\ge 0} \Psi_{\ell_p^d}(P,t)\). Then, by Lemma~\ref{lm:Median_not_to_far} and the fact that \(\|\tilde{M}_{p,1}(x-t)\|_1 = \|x-t\|_p\) we know that the diameter of \(f(Q)\) is bounded as follows
\begin{align*}
    \max_{x,y\in Q} \|f(x) -f(y)\|_1
    \le 2\max_{x\in Q}\|\tilde{M}_{p,1}(x-t)\|_1
    &= 2\max_{x\in Q}\|x-t\|_p\\
    &\lesssim \max_{x,y\in Q} \|x -y\|_p
    \lesssim \sup_{t\ge 0} \Psi_{\ell_p^d}(P,t).
\end{align*}
We can now complete the proof, as we did in the proof of Theorem~\ref{thm:avg2week}, by appealing to Lemma~\ref{lm:avg2wkavg-hard}, but with \(C\lesssim 1\) rather than \(C\lesssim \lip{f}\).
\end{proof}

We now restate and prove Theorem~\ref{thm:ds-ellp}.
\dsellp*
\begin{proof}
By standard reductions, the \((c,r)\)-NNS problem over \(\ell_p^d\) can be reduced to solving the \((c,1)\)-NNS problem for point sets in \([-\Delta, \Delta]^d\) where \(\Delta \in \poly(d)\) (see~\cite{spectral}).
Theorem~\ref{thm:wkembed-ellp} and Lemma~\ref{lm:wkvg2lsh} imply that, for a large enough \(c \lesssim \frac{p}{\varepsilon}\), and for any \((c, \frac12)\)-dispersed \(n\) point set \(P\) in \(\ell_p^d\) restricted to \([-\Delta, \Delta]^d\), there exists a \((1, p_1, p_2)\)-empirically sensitive \(\lsh(P)\) for \(P\) with \(\frac{\log(1/p_1)}{\log(1/p_2)} \le \varepsilon\) and \(\frac{1}{\ln(1/p_2)}\le \frac{1}{1-p_2} \in \poly(d)\). Moreover, a function \(h \sim \lsh(P)\) can be sampled in time \(\poly(nd)\), stored using \(\poly(d)\) bits, and evaluated in \(\poly(d)\) time. The theorem now follows from Lemma~\ref{lm:main-ds}.
\end{proof}

\subsection{Embeddings and Data Structures for Schatten-\(p\) Spaces}

The Schatten-\(p\) norms are a natural extension of the \(\ell_p\) norms to matrices. For a \(d\times e\) real matrix \(X\), and \(p \ge 1\), its Schatten-\(p\) norm \(\|X\|_{C_p}\) of \(X\) equals the \(\ell_p\) norm of its singular values, which can also be written as
\[
\|X\|_{C_p} = \mathrm{tr}( |X|^p )^{1/p},
\]
where \(|X| = (X^\top X)^{1/2}\) is the positive semi-definite square root of \(X^\top X\).
In particular, the Schatten-1 norm, sometimes called the trace norm, equals the sum of singular values, and the Schatten-\(\infty\) norm is just the maximum singular value, i.e., the spectral norm of the matrix. The Schatten-2 norm is special, as it equals the Euclidean norm of \(X\) seen as a \(de\)-dimensional vector. 

We will constrain ourselves to \(d\times d\) symmetric matrices \(X\). This is without loss of generality, since we can map any matrix \(X\) to the symmetric matrix
\[
\frac12 \begin{pmatrix}
{0} & X\\
X^\top & {0}
\end{pmatrix},
\]
where the \({0}\)'s denote the zero blocks of appropriate dimensions. This map is linear, preserves the Schatten-\(p\) norms, and produces only symmetric matrices in its image. 

For our results for Schatten-\(p\) norms, we will utilize the non-commutative Mazur maps \(M_{p,q}\). For any \(d\times d\) symmetric matrix \(X\), this map is defined by
\[ 
M_{p,q}(X) = X |X|^{\frac{p}{q}-1}.
\]
An important result of Ricard~\cite{R15} establishes that the non-commutative Mazur map satisfies similar continuity estimates as its classical commutative variant.
\begin{lemma}\label{lm:ricard}
For any \(d\times d\) symmetric matrices \(X,Y\) the following holds. If \(1 \le p \le 2\), then 
\begin{align*}
    \|M_{p,2}(X) - M_{p,2}(Y)\|_{C_2} &\lesssim \|X-Y\|_{C_p}^{p/2},\\
    \|M_{2,p}(X) - M_{2,p}(Y)\|_{C_p} &\lesssim \|X-Y\|_{C_2} (\|X\|_{C_2}^{\frac2p -1}+\|Y\|_{C_2}^{\frac2p -1}).
\end{align*}
If \(p \ge 2\), then 
\begin{align*}
    \|M_{p,2}(X) - M_{p,2}(Y)\|_{C_2} &\lesssim p\|X-Y\|_{C_p} (\|X\|_{C_2}^{\frac{p}{2} -1}+\|Y\|_{C_2}^{\frac{p}{2} -1}),\\
    \|M_{2,p}(X) - M_{2,p}(Y)\|_{C_p} &\lesssim \|X-Y\|_{C_2}^{2/p}.
\end{align*}
\end{lemma}

\subsubsection{Schatten-\(p\) for \(1 \le p \le 2\)}

We first focus on the case when \(1 \le p \le 2\), for which we can give a data structure with polynomial time pre-processing, and nearly linear space. 
Lemma~\ref{lm:ricard} implies that, for any \(p \in [1,2]\), and any \(d\times d\) symmetric matrices \(X,Y\) such that \(\|X\|_{C_p} = \|Y\|_{C_p}=1\), we have
\begin{equation}\label{eq:nc-mazur-cont}
    \|X-Y\|_{C_p} \lesssim 
    \|M_{p,2}(X) - M_{p,2}(Y)\|_{C_2} \lesssim \|X-Y\|_{C_p}^{p/2}.
\end{equation}
Note, further, that, for any matrix \(X\), \(\|M_{p,2}(X)\|_{C_2} = \|X\|_{C_p}^{p/2}\), and, moreover, for any \(t \ge 0\), \(M_{p,2}(tX) = t^{p/2} M_{p,2}(X)\). Therefore, the extension function \(\tilde{M}_{p,2}(X) = \|X\|_{C_p}^{p/2}M_{p,2}\left(\frac{X}{\|X\|_{C_p}}\right)\) defined in Lemma~\ref{lm:biholder-extension} is equal to \(M_{p,2}(X)\).

The next lemma establishes an efficient method for finding a shift satisfying the conditions in Lemma~\ref{lm:embedding-gen}. 
\begin{lemma}\label{lm:schatten-center}
Suppose that \(P\) is an \(n\)-point set of symmetric \(d\times d\) matrices and that \(p\ge 1\). Let \(q = \frac{p}{2}+1\). Let \(T\) be a matrix that minimizes the convex function \(f:\R^{d \times d} \to \R\) defined by \(f(T) = \frac1n \sum_{X\in P} \|X-T\|_{C_q}^{q}\). Then 
\[
\frac1n \sum_{X\in P}M_{p,2}(X-T) = 0.
\]
\end{lemma}
\begin{proof}
The fact that \(f\) is convex follows since \(\|X-T\|_{C_q}^{q}\) is convex in \(T\) (being a convex non-decreasing function of the convex function \(\|X-T\|_{C_q}\)), and the sum of convex functions is convex. We claim that
\begin{equation}\label{eq:grad}
\nabla f(T) = - \frac{q}{n} \sum_{X\in P}M_{p,2}(X-T),
\end{equation}
and then the lemma follows because a convex function is minimized at the set of points that have zero gradient, and this set is non-empty by Lemma~\ref{lm:center}.

To show \eqref{eq:grad}, it is enough to show that the gradient of \(\|X-T\|_{C_q}^{q}\) at \(T\) is \(-q M_{p,2}(X-T)\), for which we just need to show that \(\nabla G(Y) = q M_{p,2}(Y)\) for any \(d\times d\) symmetric matrix \(Y\) and the function \(G(Y) = \|Y\|_{C_q}^{q}\). Notice that, if we use \(\lambda(Y)\) to denote the vector of eigenvalues of \(Y\) in non-increasing order of their absolute values, then \(G(Y) = g(\lambda(Y))\) where \(g:\R^d \to \R\) is defined by \(g(y) = \|y\|_{q}^{q}\). The function \(g\) is continuous and differentiable everywhere, and satisfies \(\frac{\partial g}{\partial y_i}(y) = q \ \sign(y_i)|y_i|^{p/2}\). Let the eigendecomposition of \(Y\) be \(Y = U\Lambda U^\top\). It now follows by Corollary~3.2 in~\cite{Lewis95} that
\[
\nabla G(Y) = U \mathrm{diag}(\nabla g(\lambda(Y))) U^\top = q M_{p,2}(Y),
\]
where \(\mathrm{diag}(z)\) is the diagonal matrix with \(z\) on the diagonal.
\end{proof}

We can now state our main embedding result for Schatten-\(p\) in the range \(1 \le p \le 2\).
\begin{theorem}\label{thm:embedding-Sp}
For any $p\in [1,2]$, and any $n$-point set $P$ of symmetric \(d\times d\) matrices, there exists an embedding \(f\) of the \(\frac{p}{2}\)-snowflake of the Schatten-\(p\) norm into \(\ell_2^{d^2}\) with average distortion \(D \lesssim 1\) with respect to \(P\). If \(\max_{X\in P}\|X\|_{C_p}\le \Delta\), then \(f\) can be computed in time \(\poly(nd\Delta)\), stored in \(\poly(d\Delta)\) bits, and evaluated in \(\poly(d)\) time. 
\end{theorem}
\begin{proof}
As remarked above, giving an embedding into \(\ell_2^{d^2}\) is equivalent to giving an embedding into Schatten-2, as the two normed spaces are isometric.
The average distortion result then follows from \eqref{eq:nc-mazur-cont},  Lemma~\ref{lm:schatten-center}, and Lemma~\ref{lm:embedding-gen}. In particular, a shift \(T\) satisfying the second condition of Lemma~\ref{lm:embedding-gen} can be computed in polynomial time by minimizing the convex function from Lemma~\ref{lm:schatten-center} using a polynomial time convex minimization algorithm, e.g., the ellipsoid method~\cite{GLS12}.  
\end{proof}

The next theorem (already stated in the Introduction) gives our data structure for the NNS problem over Schatten-\(p\), \(1 \le p \le 2\).
\schattensmall*
\begin{proof}
Similarly to the proof of Theorem~\ref{thm:ds-ellp}, we can  reduce the \((cr,r)\)-NNS problem over Schatten-\(p\) to the \((c,1)\)-NNS problem for point sets of symmetric matrices whose coordinates are in \([-\Delta, \Delta]\) for \(\Delta \in \poly(d)\). Now observe that the \((c,1)\)-NNS problem over Schatten-\(p\) is equivalent to the \((c^{p/2},1)\)-NNS problem over the \(\frac{p}{2}\)-snowflake of Schatten-\(p\). The theorem now follows from Theorems~\ref{thm:avg2NNS}~and~\ref{thm:embedding-Sp}.
\end{proof}

\subsubsection{Schatten-\(p\) for \(p \ge 2\)}

Next we deal with the remaining case, \(p \ge 2\). Here, Lemma~\ref{lm:ricard} implies that there exists a constant \(C \ge 1\) such that for any \(p \ge 2\) and any \(d\times d\) symmetric matrices \(X,Y\) such that \(\|X\|_{C_p} = \|Y\|_{C_p} = 1\) we have
\begin{equation}\label{eq:nc-mazur-cont2}
    \frac{1}{C^{p/2}} \|X-Y\|_{C_2}^{p/2} \lesssim \|M_{p,2}(X) - M_{p,2}(Y)\|_{C_2} \lesssim p\|X-Y\|_{C_p} .
\end{equation}
These bounds are analogous to the ones in \eqref{eq:mazur-cont}, and suggest an analogue of Theorem~\ref{thm:embedding-ellp} for the Schatten-\(p\) norm. The idea of using coordinate-wise medians is, however, specific to the \(\ell_p\) norm and does not extend to Schatten-\(p\). Lemma~\ref{lm:schatten-center} also does not seem to extend to the case \(p\ge 2\) because it is no longer true that \(\tilde{M}_{p,2}(X) = \|X\|_{C_p}M_{p,2}\left(\frac{X}{\|X\|_{C_p}}\right)\) is equal to \(M_{p,2}(X)\), and we do not know whether \(\tilde{M}_{p,2}(X)\) can be expressed as the gradient of a convex function.  We leave the design of an efficient algorithm for computing a center satisfying the conditions of Lemma~\ref{lm:embedding-gen} for embedding Schatten-\(p\) into Schatten-2 (and, equivalently, \(\ell_2^{d^2}\)) as an open problem. Instead, we just observe that the embedding implied together by Lemma~\ref{lm:embedding-gen} and Lemma \ref{lm:center} can be computed in exponential time, and stored and evaluated efficiently. 
\begin{theorem}\label{thm:embedding-Sp-2}
For any $p\ge 2$, and any $n$-point set $P$ of symmetric \(d\times d\) matrices, there exists an embedding \(f\) of the Schatten-\(p\) norm into \(\ell_2^{d^2}\) with average distortion \(D \lesssim p\) with respect to \(P\). If \(\max_{X\in P}\|X\|_{C_p}\le \Delta\), then \(f\) can be computed in time \(2^{\poly(d\Delta)}\), stored in \(\poly(d\Delta)\) bits, and evaluated in \(\poly(d)\) time. 
\end{theorem}
\begin{proof}
A shift \(T\) satisfying the second condition of Lemma~\ref{lm:embedding-gen} exists by \eqref{eq:nc-mazur-cont2} and Lemma~\ref{lm:center}. Moreover, \(\|T\|_{C_p}\) can be bounded in terms of \(p\), and inspecing the proof of Lemma~\ref{lm:center} shows that the bound is at most exponential in \(\poly(p)\). We can assume, without loss of generality, that \(p \le \ln d\), since for any larger \(p\), Schatten-\(p\) has a bi-Lipschitz embedding into Schatten-\((\ln d)\) with constant distortion\footnote{This is a simple consequence of H\"{o}lder's inequality applied to the vector comprising of the singular values of a matrix in Schatten-$p$, and the all ones vector.}. Therefore, by discretizing a Schatten-\(p\) ball of radius in \(2^{\poly(p)}\) and enumerating the points in the discretization, we can find a shift  \(T\) satisfying the second condition of Lemma~\ref{lm:embedding-gen}  in time \(2^{\poly(d)}\).
The average distortion result then follows from Lemma~\ref{lm:embedding-gen} and the fact that the Schatten-2 norm is isometric to \(\ell_2^{d^2}\). 
\end{proof}

The next theorem (also stated in the Introduction) is an immediate corollary of Theorems~\ref{thm:avg2NNS}~and~\ref{thm:embedding-Sp-2}.
\schattenlarge*

\section{Conclusion and Open Problems}
We have constructed data structures for the $(c,r)$-NNS problem with efficient pre-processing, nearly linear space, and sub-linear query time with approximation $c\lesssim p$ in the case of $\ell_p$ spaces for all $p\geq 1$, and with $c\lesssim 1$ for Schatten-$p$ spaces for $1\leq p\leq 2$. Furthermore, we have laid out a general framework for producing such efficient data structures for general {metrics}: as long as there are (computationally efficient) average distortion embeddings of such metrics into $\ell_1^d$ or $\ell_2^d$, we can produce efficient NNS data structures (Theorem \ref{thm:avg2NNS}). This framework is an analogue of the cutting modulus framework from~\cite{spectral}, but allows efficient pre-processing. 

This connection between NNS data structures and low average distortion embeddings naturally warrants further research into constructing such computationally efficient embeddings of metric spaces into \(\ell_1^d\) or \(\ell_2^d\).  An interesting challenge is presented by the Schatten-$p$ spaces where $p>2$. As noted earlier, the bottleneck in our construction (Theorem~\ref{thm:embedding-Sp-2}) is the design of an efficient algorithm for computing a center $T$ satisfying the conditions of Lemma~\ref{lm:embedding-gen} for embedding Schatten-\(p\) into Schatten-$2$.

\begin{problem}
Given a dataset $P$ of $n$ $d\times d$ symmetric matrices, find a matrix $T$ in $\poly(n,d,1/\varepsilon)$ time such that either $0$ is a $(C,\varepsilon)$-median of $\tilde{M}_{p,2}(T-P)$, or \(\left\|\frac1n \sum_{X\in P}\tilde{M}_{p,2}(T-X)\right\|_{C_2}\le \varepsilon\).
\end{problem}

NNS data structures with approximation \(d^{o(1)}\), small space complexity, and efficient query time, but exponential preprocessing, are known for arbitrary \(d\)-dimensional norms~\cite{daher}. We conjecture that the approximation can be improved, and the preprocessing time can also be made polynomial via our low average distortion embeddings framework. It would suffice to make the main result of \cite{naor2019average} algorithmic, as follows.

\begin{problem}
Given an $n$-point dataset $P$ in a $d$-dimensional Banach space $(X,\|\cdot\|)$, construct an embedding $f:X^\frac12 \to \ell_2^d$ (where $X^{\frac12}$ is the $\frac12$-snowflake of $X$) with $2$-average distortion $\lesssim \sqrt{\log d}$ with respect to $P$ such that \(f\) can be computed from $P$ in time \(\poly(nd)\), stored using \(poly(d)\) bits, and evaluated in time \(\poly(d)\).
\end{problem}
A solution to this problem will imply NNS data structures for any \(d\)-dimensional norm with polynomial time pre-processing, nearly linear space, sub-linear query time, and approximation poly-logarithmic in the dimension, solving also an open problem in~\cite{spectral}. Note that such data structures are not known even with exponential pre-processing, but it was shown in~\cite{spectral} that they do exist in the cell-probe model.

Finally, on a somewhat different note, it would also be very interesting to further optimize the approximation factor $c$ of our NNS data structures, even in the special case of $\ell_p$ spaces. 

\begin{problem}
Establish Theorem~\ref{thm:ds-ellp} with $c\lesssim \frac{\log p}{\varepsilon}$. 
\end{problem}
A solution to this problem would interpolate between the data structures for \(\ell_1^d\) and \(\ell_2^d\) where constant approximation is possible, and Indyk's data structure for \(\ell_\infty^d\) which guarantees an \(O(\log \log d)\) approximation~\cite{I01} and is optimal in several natural models~\cite{ACP08}.

\bibliographystyle{alpha}
\bibliography{bibfile}
\end{document}